\documentclass{article}
\usepackage{amsthm}
\usepackage{amsmath}
\usepackage{amsfonts}
\usepackage{amssymb}
\usepackage{bbm}
\usepackage{enumitem}
\usepackage{geometry}
\usepackage[flushleft]{threeparttable}
\usepackage{hyperref}
\hypersetup{
    colorlinks=true,
    linkcolor=red,
    filecolor=magenta,      
    urlcolor=cyan,
    citecolor=blue,
}
\usepackage{indentfirst}
\usepackage{appendix}
\usepackage{natbib}
\usepackage{booktabs}
\usepackage{graphicx}
\usepackage{float}
\usepackage[symbol]{footmisc}
\usepackage{authblk}
\usepackage{booktabs}
\usepackage{multirow}

\usepackage{xpatch}
\makeatletter
\AtBeginDocument{\xpatchcmd{\@thm}{\thm@headpunct{.}}{\thm@headpunct{}}{}{}}
\makeatother

\graphicspath{{figures/}}
\DeclareMathOperator*{\Var}{Var}
\DeclareMathOperator*{\Cov}{Cov}

\newtheorem{theorem}{Theorem}
\newtheorem{corollary}{Corollary}
\newtheorem{lemma}{Lemma}
\newtheorem{assumption}{Assumption}
\newtheorem{example}{Example}

\title{\textbf{Subsampling Under Two-way Clustering with Serial Correlation}\thanks{We thank Harold D. Chiang and Bruce E. Hansen for their valuable advise. We are also grateful to the comments from the participants of the econometrics lunch presentation in UW-Madison and the New York Camp Econometrics XX, including Yong Cai, Ulrich Hounyo, Jack Porter, Xiaoxia Shi, Yuya Shimizu, and Kohei Yata.}}
\author{Haonan Miao}
\affil{Department of Economics, University of Wisconsin-Madison}
\date{\today}
\begin{document}
\maketitle

\begin{abstract}
    We prove the validity of using subsampling method for inference under a two-way clustered panel in which the time effects are serially correlated. Subsamples should be drawn without replacement from randomly partitioned individual index set and consecutive blocks of time effects. We present two subsampling inference methods: estimating the quantiles directly and constructing the confidence interval by first estimating the asymptotic variance. The quantile method is very adaptive, allowing for non-Gaussian limit which invalidates all existing methods in two-way clustering with serial correlation. Although the variance method only works under Gaussian limit, it comes with a data-driven bandwidth selection algorithm and a bias-correction under suitable estimators. Monte Carlo simulations demonstrate our methods exhibiting the desired coverage level in the finite sample except when the serial correlation is extremely strong. This paper is the first one that allows for inference on non-Gaussian asymptotics under two-way clustering with serial correlation.
\end{abstract}

\section{Introduction}
Inference on clustering data has always been an important subject and has been extensively studied in the past. In the realm of multiway clustering, \cite{cameron_robust_2011} is probably one of the most cited papers, which proposes multiway-clustering-robust standard errors. \cite{petersen_estimating_2009} explains when firm clustering, time clustering, or two-way clustering is appropriate. \cite{mackinnon_wild_2021} studies the wild bootstrap procedures and theoretical conditions for valid inference with two-way clustered data; \cite{menzel_bootstrap_2021} also proposes using bootstrap method for inference under multiway clustering and gives the corresponding asymptotic theory. However, none of the above studies allows for correlation within at least one dimension of the clusters, which may accompany the multiway clustering problem up to some extent. The multiway clustering is quite common in the panel data, where one dimension is usually firm or individual and the other one is time. Yet, the common time effect is typically serially correlated. Hence, this problem poses a new challenge to the inference method, which requires the researchers to account for both multiway clustering and the potential serial correlation at the same time, as \cite{bertrand_how_2004} has shown the necessity of considering dependence in practical study. \cite{chiang_standard_2024} is probably one of the first papers that takes such problem into consideration. They give a new variance estimator and the corresponding asymptotic theory under a two-way clustering framework, with allowing the serial correlation up to a certain threshold. \cite{chen_fixed-b_2024} obtains a new asymptotic result of the \cite{chiang_standard_2024} estimator (CHS estimator hereafter) under a fixed ratio between the bandwidth and the sample size, which they call it fixed-b, and proposes two bias-corrected variants of the CHS estimator along with the critical values under the fixed-b asymptotics. The most recent study on this issue is probably \cite{hounyo_reliable_2024}, which proposes two multiway wild cluster bootstrap methods based on CHS asymptotic-based approach. They show the wild bootstrap method works both theoretically and practically under a two-way clustering with serial correlation framework.\\
\par
Interestingly, as opposed to bootstrap, using subsampling for inference under two-way clustering with serial correlation remains almost unexplored. The use of subsampling as an inference method can be dated back to \cite{mahalanobis_sample_1946}, which suggests to employ subsamples to estimate variance in crop yields. \cite{carlstein_use_1986} uses non-overlapping subsamples to estimate variance in time series. The leave-many-out jackknife, which is very similar to subsampling, is studied as an approach to estimate variance by \cite{wu_jackknife_1986} and \cite{shao_general_1989}. \cite{politis_subsampling_1999} presents a very complete theory on subsampling, including basic properties of subsampling distribution under i.i.d. case, stationary and non-stationary time series, and various other settings. As subsampling works under both independent data and time series, then it seems promising that it would work under two-way clustering with serial correlation. \cite{menzel_bootstrap_2021} considers using subsampling for inference under two-way clustering where both individual and time effects are i.i.d. draws. Other recent usage of subsampling as an inference method include \cite{chernozhukov_inference_2011} and \cite{kurisu_subsampling_2025}, both of which consider extremal conditional quantile regression. Nevertheless, there is no similar result or theorem that works for the two-way clustering with serial correlation, and this paper aims to fill this gap. Using subsampling for this problem is intuitive, as each proper subsample is in fact a smaller sample and would reflect the true underlying distribution. Then, repeatedly drawing subsamples from the original sample is similar to drawing many samples from the true underlying distribution, which could not be done in reality. Therefore, we can recover some statistics of the true underlying distribution via this procedure.\\
\par
The main contribution of this paper is that it shows the validity of using subsampling methods for inference under two-way clustering with serial correlation. We prove the subsampling distribution converges to the true underlying distribution in probability pointwise under mild conditions; henceforth, its quantiles also converge to their counterparts in the true underlying distribution and can be used to construct confidence intervals with desired levels. Furthermore, such method does not require knowing the form of the limiting distribution but only its existence. The existing literature on this problem is either CLT-based (\cite{chiang_standard_2024}; \cite{chen_fixed-b_2024}) or depends on the normality of t-statistics (\cite{hounyo_reliable_2024}). However, due to a lack of subsample size selection algorithm, using quantiles of the subsample distribution for inference is limited at this stage. With an additional uniform integrability condition, we also develop a variance estimator based on the subsampling distribution, which is consistent to the variance of the true underlying distribution. Such variance estimator can be used to construct confidence intervals or test statistics as well. We further suggest a data-driven method to select subsample size by connecting this subsampling variance estimator to the well-studied spectral density estimator, when the parameter of interest is the population average. A bias-corrected version of this variance estimator that reduces the order of bias is also proposed. The simulation results agree with the theoretical argument and show our methods are comparable, in terms of the inference quality, to the existing leading method, except when the serial correlation is extremely strong.\\
\par
The paper is organized as the following. Section \ref{sec2} talks about the subsampling distribution, and  Section \ref{sec3} is about the variance estimator. The simulation result is presented in Section \ref{sec4}. Appendix \ref{AppendixA} and \ref{AppendixB} give proofs to the main theorems, whereas Appendix \ref{Append_lemma} contains proofs of technical lemmas. The derivations of some results and equalities are in Appendix \ref{Append_detail}, along with other technical details.

\section{Subsampling Distribution in Panel Data}\label{sec2}
\subsection{Quantiles of Subsampling Distribution}
We consider a panel data $\{X_{nt}:1\leq n\leq N,1\leq t\leq T\}$, where $N$ is the sample size of individuals and $T$ the sample size of time periods. We assume each $X_{nt}$ is generated via some Borel-measurable function $f$
\begin{equation}\label{general framework}
    X_{nt}=f(\alpha_{n},\gamma_{t},\varepsilon_{nt})
\end{equation}
with $\{\alpha_{n}\},\{\gamma_{t}\}$, and $\{\varepsilon_{nt}\}$ being mutually independent, $\{\alpha_{n}\}$ being i.i.d. across $n$, $\{\gamma_{t}\}$ being strictly stationary and serially correlated, and $\{\varepsilon_{nt}\}$ being i.i.d. across $(n,t)$. 
\begin{assumption}\label{a1}
    (i) $X_{nt}=f(\alpha_{n},\gamma_{t},\varepsilon_{nt})$ for some Borel measurable function $f$. (ii) $\{\alpha_{n}\},\{\gamma_{t}\}$, and $\{\varepsilon_{nt}\}$ are mutually independent. (iii) $\{\alpha_{n}\}$ is i.i.d. across $n$, and $\{\varepsilon_{nt}\}$ is i.i.d. across $(n,t)$. (iv) $\{\gamma_{t}\}$ is strictly stationary and strong mixing.
\end{assumption}
\par
Let $\theta$ be the parameter of interest. In the main body of this paper, we assume $\theta$ is real-valued, while all the theorems and results can be easily extended to the multivariate case. Define $\hat{\theta}_{NT}=\hat{\theta}_{NT}(\{X_{nt}:1\leq n\leq N,1\leq t\leq T\})$, an estimator of $\theta$ using the full sample. Furthermore, let $\tau_{NT}$ be some normalizing constant and $J_{NT}(x):=\mathbb{P}(\tau_{NT}(\hat{\theta}_{NT}-\theta)\leq x)$ to be the cumulative distribution function of the root $\tau_{NT}(\hat{\theta}_{NT}-\theta)$. Since our goal is to do inference properly on the estimator $\hat{\theta}_{NT}$, then we also need the limiting distribution of the root to exist, which is described by the following assumption.
\begin{assumption}\label{a2}
    There exists some distribution function $J$ such that as $N,T\rightarrow\infty$ $J_{NT}(x)=\mathbb{P}(\tau_{NT}(\hat{\theta}_{NT}-\theta)\leq x)\rightarrow J(x)$ $\forall x\in\mathbb{R}$ at which $J$ is continuous.
\end{assumption}
Subsampling in panel data is not merely a combination of subsampling in i.i.d. data and the one in time series data. In i.i.d. data, we usually pick $b$ indices from $\{1,\dots,N\}$ without replacement and allow overlapping subsamples, and we are able to bound tail probabilities by Hoeffding-Serfling inequality (\cite{serfling_u-statistics_1980}, Theorem A, p.201) due to independence. However, under two-way clustering with serial correlation, none of any pair of observations are independent. Yet, any two observations that do not share the same individual index and far away in time are close to independent due to the strong mixing condition on the time component $\{\gamma_{t}\}$. Therefore, to construct a subsample, we randomly partition $\{1,\dots,N\}$ into $\lceil\frac{N}{b}\rceil$ disjoint subsets with all but one having a size of $b$ and pick $l$ consecutive indices from $\{1,\dots,T\}$ to preserve the serial correlation. For instance, $\{X_{nt}:n\in I_{i},k\leq t\leq k+l-1\}$ is a subsample of individuals from $I_{i}\subset\{1,\dots,N\}$, with $|I_{i}|=b$, and starting from time $k$. In total, there are $\lceil\frac{N}{b}\rceil\cdot(T-l+1)$ possible subsamples. The corresponding subsample estimator $\hat{\theta}_{b,l,i,k}$ is defined as $\hat{\theta}_{b,l,i,k}=\hat{\theta}_{bl}(\{X_{nt}:n\in I_{i},k\leq t\leq k+l-1\})$. The merit of the subsampling is a “tautology": a subsample is in fact a smaller sample. The distribution of $\tau_{bl}(\hat{\theta}_{bl}-\theta)$ is exactly $J_{bl}$, and it has the same limiting distribution as the root $\tau_{NT}(\hat{\theta}_{NT}-\theta)$ as long as we think $b$ and $l$ as the smaller versions of $N$ and $T$ respectively. Not only both $b$ and $l$ must go to infinity, but also their relation has to be the same as the relation between $N$ and $T$ in the limit. For example, if $N$ goes to infinity faster than $T$, then $b\ll l$ would result in $J_{bl}$ and $J_{NT}$ having different limits. Another way to think this is to consider two sequences $N({m}):\mathbb{N}\rightarrow\mathbb{N}$ and $T({m}):\mathbb{N}\rightarrow\mathbb{N}$. $N$ and $b$ are both from the sequence $(N(m))$, and $T$ and $l$ are both from the sequence $(T(m))$. But for the $m_{1}$ and $m_{2}$ such that $N=N(m_{1})$, $T=T(m_{1})$ and $b=N(m_{2})$, $l=T(m_{2})$, $m_{1}>m_{2}$. With these, repeatedly drawing subsamples from the distribution $J_{bl}$ and evaluating the empirical distribution should somehow reflect the true underlying distribution $J_{NT}$ which eventually goes to $J$, the limit associated to our parameter of interest $\theta$. These lead to our first main theorem.\\
\par
Define the subsampling distribution as
\begin{equation}\label{empiricalL}
    L_{N,T,b,l}(x)=\frac{1}{N_{b}\cdot q}\sum_{i=1}^{N_{b}}\sum_{k=1}^{q}\mathbbm{1}\{\tau_{bl}(\hat{\theta}_{b,l,i,k}-\hat{\theta}_{NT})\leq x\}
\end{equation}
where $N_{b}=\lceil\frac{N}{b}\rceil$ and $q=T-l+1$.
\begin{assumption}\label{a3}
    $\frac{\tau_{bl}}{\tau_{NT}},\frac{b}{N},\frac{l}{T}, \frac{N}{T}-\frac{b}{l}\rightarrow0$, and $b,l\rightarrow\infty$ as $N,T\rightarrow\infty$.
\end{assumption}

\begin{theorem}\label{thm1}
    Under Assumption \ref{a1}, \ref{a2}, and \ref{a3}
    \begin{enumerate}[label=(\alph*)]
        \item If $J$ is continuous at $x$, then $L_{N,T,b,l}(x)\rightarrow_{p}J(x)$.
        \item If $J$ is continuous, then $\sup_{x}|L_{N,T,b,l}(x)-J(x)|\rightarrow_{p}0$.
        \item If $J$ is continuous at $\inf\{x:J(x)\geq1-\alpha\}$, then $\mathbb{P}(\tau_{NT}(\hat{\theta}_{NT}-\theta)\leq c_{b,l}^{L}(1-\alpha))\rightarrow1-\alpha$ $\forall\alpha\in(0,1)$, where $c_{b,l}^{L}(1-\alpha)=\inf\{x:L_{N,T,b,l}(x)\geq1-\alpha\}$
    \end{enumerate}
\end{theorem}
The proof is in the Appendix \ref{AppendixA}. Part (a) provides a basic pointwise convergence in probability result which is used to derive the convergence of quantile in part (c). In part (b), the pointwise convergence in probability is strengthened to an uniform convergence in probability if the limiting distribution is continuous. Part (c) provides a valid inference method using subsampling: we can use the quantile from the empirical distribution $L_{N,T,b,l}(x)$ to construct a confidence interval, of which coverage probability goes to the nominal level. The one-sided confidence interval $(\hat{\theta}_{NT}-\frac{c_{b,l}^{L}(1-\alpha)}{\tau_{NT}},\infty)$ has an asymptotic coverage probability being the nominal level $1-\alpha$. To construct a two-sided equal-tail confidence interval with a correct asymptotic coverage probability, one should use $(\hat{\theta}_{NT}-\frac{c_{b,l}^{L}(1-\frac{\alpha}{2})}{\tau_{NT}},\hat{\theta}_{NT}-\frac{c_{b,l}^{L}(\frac{\alpha}{2})}{\tau_{NT}})$, where left and right endpoints are interchanged like the bootstrap percentile-t interval. This subsampling quantile method does not require any specific form of the limiting distribution, whereas the existing literature on two-way clustering with serial correlation are built on the Gaussianity of the asymptotic $J$. This allows us to apply subsampling to some non-standard case. 
\begin{example}[Non-separable Heterogeneity]
    This is Example 1.7 in \cite{menzel_bootstrap_2021}. Suppose
    \begin{equation*}
        X_{nt}=\alpha_{n}\gamma_{t}+\varepsilon_{nt}
    \end{equation*}
    with $\mathbb{E}[\alpha_{n}]=\mathbb{E}[\gamma_{t}]=\mathbb{E}[\varepsilon_{nt}]=0$. Beyond Assumption \ref{a1}, if we further assume $\mathbb{E}[|X_{nt}|^{2+\delta}]<\infty$ for some $\delta>0$ and $\sum_{m=0}^{\infty}\alpha_{\gamma}(m)^{1-\frac{2}{2+\delta}}<\infty$, then $\frac{1}{\sqrt{NT}}\sum_{n=1}^{N}\sum_{t=1}^{T}X_{nt}\rightarrow_{d}\sigma_{\alpha}\sigma_{\gamma}Z_{1}Z_{2}+\sigma_{\varepsilon}Z_{3}$ for some $\sigma_{\alpha},\sigma_{\gamma},\sigma_{\varepsilon}$ and mutually independent standard normal random variables $Z_{1},Z_{2},Z_{3}$. The limiting distribution is not normal but a summation of two scaled chi-squared distributions and a normal distribution. Furthermore, the convergence rate is $\sqrt{NT}$ rather than $\sqrt{N}$ or $\sqrt{T}$, which indicates it is also a degenerate case.
\end{example}
Although the knowledge of the exact form of $J$ is not required, knowing the normalizing constant $\tau_{NT}$ is still necessary to construct valid confidence interval. That is, such method works in the degenerate case when it is known to be degenerate.

\section{Variance Estimation}\label{sec3}
\subsection{Subsampling Variance Estimator}
In some cases, we can also do inference on $\hat{\theta}_{NT}$ properly by estimating its asymptotic variance. For example, if the limiting distribution is Gaussian and has an unknown variance $V$, then $\hat{\theta}_{NT}$ approximately follows the distribution of $N(\theta,\frac{V}{\tau_{NT}^{2}})$. Furthermore, whenever we want to evaluate the efficiency or derive the mean squared-error of the estimator $\hat{\theta}_{NT}$, a variance estimator is usually required. We here provide a subsampling method to estimate the variance consistently by using the variance associate to the empirical distribution $L_{N,T,b,l}$ defined in (\ref{empiricalL}).
\begin{theorem}\label{thm2}
    Let $\lim\Var(\tau_{NT}\hat{\theta}_{NT})=V<\infty$. Under Assumption \ref{a1} and \ref{a3}, if $\{\tau_{NT}^{4}(\hat{\theta}_{NT}-\mathbb{E}[\hat{\theta}_{NT}])^{4}\}$ is uniformly integrable, then
    \begin{equation}\label{sighat}
    \hat{\sigma}_{N,T,b,l}^{2}=\frac{\tau_{bl}^{2}}{N_{b}\cdot q}\sum_{i=1}^{N_{b}}\sum_{k=1}^{q}(\hat{\theta}_{b,l,i,k}-\Bar{\theta}_{N,T,b,l})^{2}\rightarrow_{L^{2}}V
    \end{equation}
    where $\Bar{\theta}_{N,T,b,l}=\frac{1}{N_{b}\cdot q}\sum_{i=1}^{N_{b}}\sum_{k=1}^{q}\hat{\theta}_{b,l,i,k}$.
\end{theorem}
The proof is in Appendix \ref{AppendixB}. Note that $\hat{\sigma}_{N,T,b,l}^{2}$ estimates $\Var(\tau_{NT}\hat{\theta}_{NT})$. To estimate the variance of $\hat{\theta}_{NT}$, we simply divide $\hat{\sigma}_{N,T,b,l}^{2}$ by $\tau_{NT}^{2}$. The condition $\frac{\tau_{bl}}{\tau_{NT}}$, which lies in Assumption \ref{a3}, is in fact redundant. Throughout the proof, the only two requirements put on the normalizing sequence $\{\tau_{NT}$ are $\Var(\tau_{NT}\hat{\theta}_{NT})$ converging to some finite $V$ and the fourth order uniform integrability on $\{\tau_{NT}(\hat{\theta}_{NT}-\mathbb{E}[\hat{\theta}_{NT}])\}$. Hence, if $\hat{\theta}_{NT}$ is consistent to some deterministic $\theta$, then $\hat{\sigma}_{N,T,b,l}^{2}\rightarrow_{p}0$. Although Theorem \ref{thm2} does not explicitly rely on Assumption \ref{a2}, knowing which family of distribution $\lim\tau_{NT}(\hat{\theta}_{NT}-\theta)$ belongs to is required if we want to construct confidence intervals using this variance estimator.

\subsection{Choice of Subsample Size}
Up to this point, asymptotic theory only requires $b=o(N)$ and $l=o(T)$ but does not tell us which exact number to pick. Selecting the subsample size is quite important as it would largely affect the result of our variance estimator just as the bandwidth selection in many nonparametric estimation problems. Traditionally, researchers used some sort of mean squared-error to determine the optimal bandwidth choice. In this paper, we aim to minimize the asymptotic mean squared-error (AMSE) of $\hat{\sigma}_{N,T,b,l}^{2}$, which is defined as
\begin{equation*}
    AMSE(\hat{\sigma}_{N,T,b,l}^{2},V)=Bias_{\infty}^{2}(\hat{\sigma}_{N,T,b,l}^{2},V)+Var_{\infty}(\hat{\sigma}_{N,T,b,l}^{2})
\end{equation*}
where $Bias_{\infty}(\hat{\sigma}_{N,T,b,l}^{2},V)$ and $\Var_{\infty}(\hat{\sigma}_{N,T,b,l}^{2})$ are defined as the leading terms of $Bias(\hat{\sigma}_{N,T,b,l}^{2},V)$ and $\Var(\hat{\sigma}_{N,T,b,l}^{2})$ respectively. However, either $Bias(\hat{\sigma}_{N,T,b,l}^{2},V)$ or $\Var(\hat{\sigma}_{N,T,b,l}^{2})$ depends on the specific form of $\hat{\theta}_{NT}$. Therefore, we focus on mean estimation throughout the rest of this section (and the next section).\\
\par
Under the data generating process of (\ref{general framework}) and Assumption \ref{a1}, we are interested in $\theta=\mathbb{E}[X_{nt}]$, and we assume $\theta=0$ without loss of generality. Let $\hat{\theta}_{NT}=\frac{1}{NT}\sum_{n=1}^{N}\sum_{t=1}^{T}X_{nt}$ be an estimator to $\theta$. Furthermore, we use a projection technique to transform $X_{nt}$ into a linear combination of some random variables. Define $a_{n}=\mathbb{E}[X_{nt}|\alpha_{n}]$, $b_{t}=\mathbb{E}[X_{nt}|\gamma_{t}]$, and $e_{nt}=X_{nt}-a_{n}-b_{t}$. Then
\begin{equation}\label{projection}
    X_{nt}=a_{n}+b_{t}+e_{nt}
\end{equation}
The random variables $\{a_{n}\}$, $\{b_{t}\}$, and $\{e_{nt}\}$ satisfy the following: (i) $\{a_{n}\}$ is i.i.d., and $\{b_{t}\}$ is strictly stationary; (ii) $\{e_{nt}\}$ is identically distributed across $(n,t)$; (iii) $\mathbb{E}[a_{n}]=\mathbb{E}[b_{t}]=\mathbb{E}[e_{nt}]=0$; (iv) $\{a_{n}\}$ and $\{b_{t}\}$ are independent, and $\{a_{n}\}$, $\{b_{t}\}$, and $\{e_{nt}\}$ are mutually uncorrelated; (v) $e_{nt}$ and $e_{mp}$ are independent conditional on $(\gamma_{t},\gamma_{p})$ for any $n\neq m$ and any $t,p\in\mathbb{N}$. Furthermore, define the follow variance and autocovariances
\begin{align*}
    V_{a}&=\Var(a_{1})\\
    R_{b}(k)&=\Cov(b_{1},b_{1+k})\quad V_{b}=R_{b}(0)\\
    R_{e}(k)&=\Cov(e_{11},e_{1,1+k})\quad V_{e}=R_{e}(0)\\
\end{align*}
In \cite{chiang_standard_2024}, the authors have established the asymptotic theory for the mean estimator $\hat{\theta}_{NT}$ under the framework (\ref{projection}) and the assumptions below. 
\begin{assumption}\label{a4}
    There exists some $r>1$ and $\delta>0$ such that (i) $\mathbb{E}[|X|^{4(r+\delta)}]<\infty$ and (ii) $\{\gamma_{t}\}$ is strong mixing with size $\frac{2r}{r-1}$.
\end{assumption}
\begin{assumption}\label{a5}
    $\frac{N}{T}\rightarrow c\in(0,\infty)$ as $N,T\rightarrow\infty$.
\end{assumption}
Then, under Assumption \ref{a1}, \ref{a4}, and \ref{a5},
    \begin{equation*}
        \sqrt{N}(\hat{\theta}_{NT}-\theta)\rightarrow_{d}N(0,V)
    \end{equation*}
where $V=V_{a}+c\sum_{k=-\infty}^{\infty}R_{b}(k)<\infty$ (see the derivation of $V$ in Appendix \ref{Append_detail1}).
\par
Now, we are ready to introduce a data-driven approach for selecting the subsample size in the context of mean estimation, which has a close relation to the bandwidth choice in the classical spectral density estimation. To see this relationship, note that the limiting variance $V$ can also be written as
\begin{equation*}
    V=V_{a}+c\cdot2\pi SP_{b}(0)
\end{equation*}
where $SP_{b}(\lambda)\equiv\frac{1}{2\pi}\sum_{k=-\infty}^{\infty}R_{b}(k)e^{-ik\lambda}$ is the spectral density function. Meanwhile, as the subsampling variance estimator $\hat{\sigma}_{N,T,b,l}^{2}$ is consistent and has the following expectation
\begin{equation*}
    \mathbb{E}[\hat{\sigma}_{N,T,b,l}^{2}]=V_{a}+\frac{b}{l}\sum_{k=-l+1}^{l-1}(1-\frac{|k|}{l})R_{b}(k)+\frac{1}{l}\sum_{k=-l+1}^{l-1}(1-\frac{|k|}{l})R_{e}(k)+O(\frac{b}{N})+O(\frac{bl}{N^{2}})
\end{equation*}
then we have the approximation
\begin{equation}\label{eqapprox}
    \hat{\sigma}_{N,T,b,l}^{2}\simeq V_{a}+\frac{b}{l}\sum_{k=-l+1}^{l-1}(1-\frac{|k|}{l})R_{b}(k)+\frac{1}{l}\sum_{k=-l+1}^{l-1}(1-\frac{|k|}{l})R_{e}(k)
\end{equation}
Assumption \ref{a3} says in order for $\hat{\sigma}_{N,T,b,l}^{2}$ to be a consistent estimator of V, $\frac{N}{T}-\frac{b}{l}\rightarrow0$. Then, we may assume the ratio between $b$ and $l$ is fixed and $b=\frac{N}{T}\cdot l$. Hence, there is only one parameter to choose and (\ref{eqapprox}) becomes
\begin{equation}\label{eqapprox2}
    \hat{\sigma}_{N,T,b,l}^{2}\simeq V_{a}+\frac{N}{T}\sum_{k=-l+1}^{l-1}(1-\frac{|k|}{l})R_{b}(k)+\frac{1}{l}\sum_{k=-l+1}^{l-1}(1-\frac{|k|}{l})R_{e}(k)
\end{equation}
This expression connects the subsampling variance estimator to the spectrum estimation at frequency zero using the Bartlett window. Not surprisingly, the subsample size $l$ is closely related to the Bartlett window's width; in particular, $l$ is exactly the inverse of the Bartlett window's width if we also treat $\frac{1}{l}\sum_{k=-l+1}^{l-1}(1-\frac{|k|}{l})R_{e}(k)$ negligible. At this stage, we treat this part negligible, while it would be valuable in the future if an algorithm that does not neglect such part is devised. To find the optimal subsample size, We will employ a two-step procedure to choose $l$. First, we try to identify and estimate $V_{b}$, $V_{e}$, $R_{b}$, and $R_{e}$ by some consistent estimators $\hat{R}_{b}$ and $\hat{R}_{e}$ respectively. Then, we plug in those consistent autocovariance estimators and compute $l_{opt}$ minimizing the asymptotic mean squared error. It is important to note that $\hat{R}_{b}$ and $\hat{R}_{e}$, as well as $\hat{V}_{a}$, can also be used to estimate $V$. However, we here only use them to compute the optimal subsample size. \\
\par
If we treat $\frac{\hat{V}_{e}}{l}+\frac{2}{l}\sum_{k=1}^{l-1}(1-\frac{k}{l})\hat{R}_{e}(k)$ negligible, then
\begin{equation}\label{simpleapprox}
    \hat{\sigma}_{N,T,b,l}^{2}\simeq V_{a}+\frac{N}{T}\sum_{k=-l+1}^{l-1}(1-\frac{|k|}{l})R_{b}(k)
\end{equation}
and the problem boils down to the optimal bandwidth choice of the Bartlett kernel, which has been widely studied in the past. For instance, \cite{andrews_heteroskedasticity_1991} obtained the optimal bandwidths for a class of kernels by using an asymptotic truncated mean squared error criterion; \cite{newey_automatic_1994} proposed a plug-in procedure for selecting the bandwidth. In this paper, we adopt an iterative plug-in method developed by \cite{buhlmann_locally_1996}, which was also used by \cite{buhlmann_block_1999} to select the block length of the moving blocks bootstrap in time series. To see how it connects to the optimal bandwidth choice of the Bartlett kernel, recall that $V=V_{a}+c\cdot2\pi SP_{b}(0)$. $V_{a}$ here can be thought as a constant, as a typical estimator $\hat{V}_{a}$ usually does not depend on the bandwidth parameter. Given the limiting ratio $c=\lim\frac{N}{T}$ cannot be inferred anyway, our goal becomes to select an $l$ that minimizes $AMSE(\sum_{k=-l+1}^{l-1}(1-\frac{|k|}{l})\hat{R}_{b}(k),\sum_{k=-\infty}^{\infty}R_{b}(k))$, or equivalently to select a $w$ that minimizes
\begin{equation}\label{bartlett}
    AMSE(\sum_{k=-T+1}^{T-1}W_{B}(kw)\hat{R}_{b}(k),\sum_{k=-\infty}^{\infty}R_{b}(k))
\end{equation}
where $W_{B}(kw)=\max\{0,1-|kw|\}$ is the Bartlett window, with the bandwidth $w$ being the inverse of $l$. The optimal $w$ that minimizes the problem (\ref{bartlett}) is 
\begin{equation*}
    w^{*}=[\frac{2(\sum_{k=-\infty}^{\infty}R_{b}(k))^{2}}{3(\sum_{k=-\infty}^{\infty}|k|R_{b}(k))^{2}}]^{\frac{1}{3}}T^{-\frac{1}{3}}
\end{equation*}
To estimate $w^{*}$, one usually needs to use other types of kernels to estimate $\sum_{k=-\infty}^{\infty}R_{b}(k)$ and $\sum_{k=-\infty}^{\infty}|k|R_{b}(k)$.\\
\par
Our two-step procedure will then be estimating $R_{b}(k)$ by $\hat{R}_{b}(k)$ and applying B\"uhlmann's iterative scheme to find the optimal subsample size $l_{opt}$. First, to estimate $R_{b}(k)=\mathbb{E}[b_{t}b_{t+k}]$, by the properties of $\{a_{n}\}$, $\{b_{t}\}$, and $\{e_{nt}\}$, for any $n\neq m$ and $k\in\{0\}\cup\mathbb{N}$, we have
\begin{equation*}
    \mathbb{E}[X_{nt}X_{m,t+k}]=\mathbb{E}[b_{t}b_{t+k}]\\
\end{equation*}
Then, the estimator $\hat{R}_{b}(k)=\frac{1}{N(N-1)T}\sum_{n\neq m}\sum_{t=1}^{T-k}X_{nt}X_{m,t+k}$ consistently estimates $\mathbb{E}[X_{nt}X_{m,t+k}]$. We divide the time sum by $T$ instead of $T-k$ in order to decrease the volatility. When $k$ is close to $T$, there are fewer available data points for $X_{nt}X_{m,t+k}$. Hence, the quantity is very sensitive to the sample realization. By dividing $T$ instead of $T-k$, the influence of $X_{nt}X_{m,t+k}$ is decreased when $k$ is large.
The iterative steps are as follows
\begin{align}\label{lopt}
    w_{0}&=T^{-1}\notag\\
    w_{i}&=[\frac{\sum_{k=-T+1}^{T-1}\hat{R}_{b}^{2}(k)}{6\sum_{k=-T+1}^{T-1}W_{SC}^{2}(kw_{i-1}T^{4/21})k^{2}\hat{R}_{b}^{2}(k)})]^{\frac{1}{3}}T^{-\frac{1}{3}}\quad \text{for }i=\{1,\dots,L\}\notag\\
    w_{opt}&=[\frac{2(\sum_{k=-T+1}^{T-1}W_{TH}(kw_{L}T^{4/21})\hat{R}_{b}(k))^{2}}{3(\sum_{k=-T+1}^{T-1}W_{SC}(kw_{L}T^{4/21})|k|\hat{R}_{b}(k))^{2}}]^{\frac{1}{3}}T^{-\frac{1}{3}}\notag\\
    l_{opt}&=\lfloor w_{opt}^{-1}\rceil\quad\text{closest integer to }w_{opt}^{-1}
\end{align}
where
\begin{align*}
W_{TH}(x)&=
\begin{cases}
    \frac{1+\cos(\pi x)}{2} & |x|\leq1\\
    0 & |x|>1
\end{cases}\\
\\
W_{SC}(x)&=
\begin{cases}
    1 & |x|>\frac{4}{5}\\
    \frac{1+\cos(\pi(5x-4))}{2} & \frac{4}{5}\leq |x|\leq 1\\
    0 & |x|>1
\end{cases}
\end{align*}
The number $L$ is the number of iterations of the “global steps". According to \cite{buhlmann_locally_1996}, $L$ needs to be no less than 4 in order for the optimal bandwidth $w_{opt}$ to have the correct asymptotic order. In the later simulation section, we set $L=20$ to guarantee a convergence.\\
\par
The optimality of such algorithm depends on a set of assumptions. First, we need (\ref{simpleapprox}) to be a good approximation with estimating $R_{b}(k)$ by $\hat{R}_{b}(k)=\frac{1}{N(N-1)T}\sum_{n\neq m}\sum_{t=1}^{T-k}X_{nt}X_{m,t+k}$, namely,
\begin{equation*}
    \frac{\mathbb{E}[(\hat{\sigma}_{N,T,b,l}^{2}-V)^{2}]}{\mathbb{E}[(V_{a}+\frac{N}{T}\sum_{k=-l+1}^{l-1}(1-\frac{|k|}{l})\hat{R}_{b}(k)-V)^{2}]}\rightarrow1
\end{equation*}
as $N,T\rightarrow\infty$. The spectral density at frequency 0 also has to be positive, $SP_{b}(0)>0$. Furthermore, we need $\sum_{k=0}^{\infty}(k+1)^{4}|R_{b}(k)|<\infty$, and the cumulants of $\{b_{t}\}$ with order $h\leq8$ are summable, 
\begin{equation*}
    \sum_{{t_{1},\dots,t_{h-1}}=0}^{\infty}|cum_{h}(b_{0},b_{t_{1}},\dots,b_{t_{h-1}})|<\infty
\end{equation*}
With these assumptions, $l_{opt}=l^{*}(1+O(T^{-\frac{2}{7}}))$ as $N,T\rightarrow\infty$, where $l^{*}=\arg\min_{l\in\mathbb{N}} AMSE(\hat{\sigma}_{N,T,b,l}^{2},V)$.

\subsection{Bias-Correction on the Variance Estimator}
In the last section, we have pinned down the optimal subsample sizes. However, such choice does not eliminate the bias in general due to bias-variance trade-off. In this section, we will introduce a bias-corrected subsampling variance estimator that reduces the order of bias under the framework (\ref{projection}) and Assumption \ref{a1}-\ref{a5}.\\
\par
Recall that 
\begin{equation*}
    \mathbb{E}[\hat{\sigma}_{N,T,b,l}^{2}]=V_{a}+\frac{b}{l}\sum_{k=-l+1}^{l-1}(1-\frac{|k|}{l})R_{b}(k)+\frac{1}{l}\sum_{k=-l+1}^{l-1}(1-\frac{|k|}{l})R_{e}(k)+O(\frac{b}{N})+O(\frac{bl}{N^{2}})
\end{equation*}
then the bias is
\begin{align}\label{bias}
    Bias(\hat{\sigma}_{N,T,b,l}^{2},V)&=\mathbb{E}[\hat{\sigma}_{N,T,b,l}^{2}]-V\notag\\
    &=-\frac{b}{l^{2}}\sum_{k=-\infty}^{\infty}|k|R_{b}(k)+(\frac{b}{l}-c)\sum_{k=-\infty}^{\infty}R_{b}(k)+2\frac{b}{l}\sum_{k=l}^{\infty}(\frac{k}{l}-1)R_{b}(k)\notag\\
    &+\frac{1}{l}\sum_{k=-l+1}^{l-1}(1-\frac{|k|}{l})R_{e}(k)+O(\frac{b}{N})+O(\frac{bl}{N^{2}})
\end{align}
The detail is in Appendix \ref{Append_detail3}. The term $(\frac{b}{l}-c)\sum_{k=-\infty}^{\infty}R_{b}(k)$ involves the unknown constant $c=\lim\frac{N}{T}$, then for simplicity of exposition, we assume $\frac{N}{T}$ is close enough to $c$ such that $(\frac{N}{T}-c)\sum_{k=-\infty}^{\infty}R_{b}(k)$ is in a small order. With $l=l_{opt}$ defined in (\ref{lopt}) and $b=\frac{N}{T}\cdot l$, $b,l=O(T^{\frac{1}{3}})=O(N^{\frac{1}{3}})$, and (\ref{bias}) simplifies to
\begin{equation}\label{bias-simple}
    Bias(\hat{\sigma}_{N,T,b,l}^{2},V)=-\frac{1}{l}\sum_{k=-\infty}^{\infty}c|k|R_{b}(k)+\frac{1}{l}\sum_{k=-l+1}^{l-1}(1-\frac{|k|}{l})R_{e}(k)+o(\frac{1}{l})
\end{equation}
Due to the finiteness of $\sum_{k=-\infty}^{\infty}|k|R_{b}(k)$ (see Appendix \ref{Append_detail2}), $Bias(\hat{\sigma}_{N,T,b,l}^{2},V)=O(\frac{1}{l})$. This inspires the following bias-corrected variance estimator
\begin{equation}\label{bcsighat}
    \hat{\sigma}_{N,T,b,l}^{2,BC}=\hat{\sigma}_{N,T,b,l}^{2}-D\cdot(\hat{\sigma}_{N,T,\Tilde{b},\Tilde{l}}^{2}-\hat{\sigma}_{N,T,b,l}^{2})
\end{equation}
for some constant $D$, where $\Tilde{b}\ll b$, $\Tilde{l}\ll l$ and $\Tilde{b},\Tilde{l}\rightarrow\infty$ as $N,T\rightarrow\infty$. The idea is that we use $\hat{\sigma}_{N,T,\Tilde{b},\Tilde{l}}^{2}$, which is the same type as $\hat{\sigma}_{N,T,b,l}^{2}$ but with even fewer subsamples, to estimate $\mathbb{E}[\hat{\sigma}_{N,T,b,l}^{2}]$. Therefore $\hat{\sigma}_{N,T,\Tilde{b},\Tilde{l}}^{2}-\hat{\sigma}_{N,T,b,l}^{2}$ estimates $Bias(\hat{\sigma}_{N,T,b,l}^{2},V)$, as $\hat{\sigma}_{N,T,b,l}^{2}$ is already an estimator to $V$. And by choosing D wisely, we can eliminate the leading terms in (\ref{bias-simple}) so that $Bias(\hat{\sigma}_{N,T,b,l}^{2,BC},V)$ has a smaller order compared to $Bias(\hat{\sigma}_{N,T,b,l}^{2},V)$. If we choose $D=\frac{\Tilde{l}}{l-\Tilde{l}}$, then we have
\begin{align*}
    Bias(\hat{\sigma}_{N,T,b,l}^{2,BC},V)&=\frac{2}{l-\Tilde{l}}\sum_{k=\Tilde{l}}^{l-1}R_{e}(k)-\frac{1}{l-\Tilde{l}}\sum_{k=-l+1}^{l-1}\frac{|k|}{l}R_{e}(k)+\frac{1}{l-\Tilde{l}}\sum_{k=-\Tilde{l}+1}^{\Tilde{l}-1}\frac{|k|}{\Tilde{l}}R_{e}(k)+o(\frac{1}{l})\\
    &=o(\frac{1}{l})
\end{align*}
\begin{corollary}\label{co3}
    Under the framework (\ref{projection}) and Assumption \ref{a1} and \ref{a3}-\ref{a5}, if $\{\tau_{NT}^{4}(\hat{\theta}_{NT}-\mathbb{E}[\hat{\theta}_{NT}])^{4}\}$ is uniformly integrable, then
    \begin{equation*}
        \hat{\sigma}_{N,T,b,l}^{2,BC}\rightarrow_{p}V
    \end{equation*}
    where $\hat{\sigma}_{N,T,b,l}^{2,BC}$ is defined in (\ref{bcsighat}).
\end{corollary}

\subsection{Linear Regression}
We consider the following linear regression
\begin{align}\label{eqlr}
    (Y_{nt},X_{nt}',U_{nt})'&=f(\alpha_{n},\gamma_{t},\epsilon_{nt})\notag\\
    Y_{nt}&=X_{nt}'\beta+U_{nt}\notag\\
    \mathbb{E}[U_{nt}|\mathbf{X}]&=0
\end{align}
Define $a_{n}=\mathbb{E}[X_{nt}U_{nt}|\alpha_{n}]$, $b_{t}=\mathbb{E}[X_{nt}U_{nt}|\gamma_{t}]$, and $e_{nt}=X_{nt}U_{nt}-a_{n}-b_{t}$. Suppose we are interested in the OLS estimator $\hat{\beta}_{OLS}=(\frac{1}{NT }\sum_{n=1}^{N}\sum_{t=1}^{T}X_{nt}X_{nt}')^{-1}(\frac{1}{NT}\sum_{n=1}^{N}\sum_{t=1}^{T}X_{nt}Y_{nt})$ and want to estimate its asymptotic variance. Instead of subsampling the whole $\hat{\beta}_{OLS}$ and computing the variance of its subsample distribution, we can subsample the score $\frac{1}{NT}\sum_{n=1}^{N}\sum_{t=1}^{T}X_{nt}U_{nt}$ and compute its subsample distribution variance $\hat{\Sigma}_{N,T,b,l}$ first, and then estimate the asymptotic variance of $\hat{\beta}_{OLS}$ by $\hat{\phi}_{NT}^{-1}\hat{\Sigma}_{N,T,b,l}\hat{\phi}_{NT}^{-1}$, where $\hat{\phi}_{NT}=\frac{1}{NT }\sum_{n=1}^{N}\sum_{t=1}^{T}X_{nt}X_{nt}'$. By doing so, we transform such problem into a mean-estimation problem and allow ourselves to apply the subsample size selection algorithm and bias-correction technique described previously.\\
\par
The problem of this approach is that $U_{nt}$ is unobservable; hence we might need to replace $X_{nt}U_{nt}$ by $X_{nt}\hat{U}_{nt}$ where $\hat{U}_{nt}=Y_{nt}-X_{nt}'\hat{\beta}_{OLS}$. Let $\check{\Sigma}_{N,T,b,l}$ be the variance of the subsample distribution associate to the practical score $\frac{1}{NT}\sum_{n=1}^{N}\sum_{t=1}^{T}X_{nt}\hat{U}_{nt}$. It turns out that the consistency of $\check{\Sigma}_{N,T,b,l}$ does not require any additional assumption other than the ones securing the asymptotic normality of $\hat{\beta}_{OLS}$ and the consistency of $\hat{\Sigma}_{N,T,b,l}$.
\begin{assumption}\label{a6}
    For some $r>1$ and $\delta>0$: (i)$\{(Y_{nt},X_{nt}',U_{nt}):1\leq n\leq N,1\leq t\leq T\}$ are generated as (\ref{eqlr}), with $\{\alpha_{n}\},\{\gamma_{t}\},\{\varepsilon_{nt}\}$ being mutually independent, $\{\alpha_{n}\}$ being i.i.d. across $n$, and $\{\varepsilon_{nt}\}$ being i.i.d. across $(n,t)$. (ii) $\{\gamma_{t}\}$ is strictly stationary and strong mixing with size $\frac{2r}{r-1}$. (iii) $\phi=\mathbb{E}[X_{nt}X_{nt}']>0$, $\mathbb{E}[||X_{nt}||^{8(r+\delta)}]<\infty$, $\mathbb{E}[||U_{nt}||^{8(r+\delta)}]<\infty$. (iv) Either $a_{n}$ or $b_{t}$ is non-degenerate. (v) $\frac{\tau_{bl}}{\tau_{NT}},\frac{b}{\sqrt{N}},\frac{l}{\sqrt{T}}\rightarrow0$, $\frac{N}{T}-\frac{b}{l}\rightarrow0$, and $b,l\rightarrow\infty$ as $N,T\rightarrow\infty$. (vi) $\{b^{2}(\hat{\theta}_{NT}-\theta)^{4}\}$ is uniformly integrable.
\end{assumption}

\begin{theorem}\label{thm3}
    Under Assumption \ref{a6}, $\hat{t}:=\hat{V}_{N,T,b,l}^{-\frac{1}{2}}\sqrt{N}(\hat{\beta}-\beta)\rightarrow_{d}N(0,1)$ and $\check{t}:=\check{V}_{N,T,b,l}^{-\frac{1}{2}}\sqrt{N}(\hat{\beta}-\beta)\rightarrow_{d}N(0,1)$, where
    \begin{align*}
        \hat{V}_{N,T,b,l}&=\hat{\phi}_{NT}^{-1}\hat{\Sigma}_{N,T,b,l}\hat{\phi}_{NT}^{-1}\\
        \check{V}_{N,T,b,l}&=\hat{\phi}_{NT}^{-1}\check{\Sigma}_{N,T,b,l}\hat{\phi}_{NT}^{-1}
    \end{align*}
\end{theorem}
The proof is in Appendix \ref{AppendixC}. (i)-(iv) of Assumption \ref{a6} are used to established the asymptotic normality of $\sqrt{N}(\hat{\beta}-\beta)$, which is obtained by \cite{chiang_standard_2024}. (v)-(vi) of Assumption \ref{a6}, together with the asymptotic normality of $\sqrt{N}(\hat{\beta}-\beta)$, establish the consistency of $\hat{\Sigma}_{N,T,b,l}$ and $\check{\Sigma}_{N,T,b,l}$ and the asymptotic normality of the corresponding t-statistics.

\section{Simulation}\label{sec4}
In this section, we evaluate the validity of the subsampling methods and compare them with other existing methods by simulations. We first consider a model with Gaussian limit in which both subsampling methods and other two-way clustering with serial correlation methods work. We also consider a degenerate and non-Gaussian limit model, in which only the subsampling quantile method should be able to provide the correct asymptotic coverage.

\subsection{Non-degenerate Gaussian Limit}
For $\{(Y_{nt},X_{nt},U_{nt}):1\leq n\leq N,1\leq t\leq T\}$, consider a linear regression
\begin{align*}
    Y_{nt}&=\beta_{0}+\beta_{1}\cdot X_{nt}+U_{nt}\\
    &=Z_{nt}'\beta+U_{nt}
\end{align*}
where $\beta_{0}=\beta_{1}=1$ and $Z_{nt}=[1\quad X_{nt}]'$. The regressor $X_{nt}$ and error $U_{nt}$ are generated as
\begin{align*}
    X_{nt}&=0.3\cdot\alpha_{n}^{x}+0.5\cdot\gamma_{t}^{x}+0.2\cdot\varepsilon_{nt}^{x}\\
    U_{nt}&=0.3\cdot\alpha_{n}^{u}+0.5\cdot\gamma_{t}^{u}+0.2\cdot\varepsilon_{nt}^{u}
\end{align*}
with $(\alpha_{n}^{x},\varepsilon_{nt}^{x},\alpha_{n}^{u},\varepsilon_{nt}^{u})$ being mutually independent standard normal random variables. The time effects $(\gamma_{t}^{x},\gamma_{t}^{u})$ are independently generated via AR(1) process as
\begin{align*}
    \gamma_{t+1}^{x}&=\rho\cdot\gamma_{t}^{x}+v_{t}^{x}\\
    \gamma_{t+1}^{u}&=\rho\cdot\gamma_{t}^{u}+v_{t}^{u}
\end{align*}
where the innovation terms $v_{t}^{x}$ and $v_{t}^{u}$ follow the distribution $N(0,1-\rho^{2})$ and are independent to $\gamma_{t+1}^{x}$ and $\gamma_{t+1}^{u}$ respectively. We choose $\rho\in\{0,0.25,0.5,0.75\}$ to see how our methods perform under different levels of serial correlation. Note that $\rho=0$ represents the case of two-way clustering without serial correlation. And our goal is to test the 95\% confidence interval coverage probability of the OLS estimator $\hat{\beta}_{1}$ using the quantiles of the subsampling distribution and the subsampling variance estimator.\\
\par
For the quantile method, we draw $N_{b}\cdot q$ many subsamples and estimate $\beta_{1}$ using each subsample by OLS. Then, we find the 2.5\% and 97.5\% quantile of the subsample OLS estimators $\{\hat{\beta}_{1,b,l,i,k}:1\leq i\leq N_{b},1\leq k\leq q\}$ and denote them as $c_{b,l}^{L}(0.025)$ and $c_{b,l}^{L}(0.975)$ respectively. The 95\% confidence interval is constructed as $[\hat{\beta}_{1}-\frac{c_{b,l}^{L}(0.975)}{\sqrt{N}},\quad\hat{\beta}_{1}-\frac{c_{b,l}^{L}(0.025)}{\sqrt{N}}]$, where $\hat{\beta}_{1}$ is the OLS estimator of $\beta_{1}$ using the full sample. Table \ref{tb1} shows the coverage probabilities for $\beta_{1}$ with a nominal probability of $95\%$ and 1000 Monte Carlo repetitions under different dependence levels. We vary $(N,T)$ from 50 to 400 for different choices of $\rho$. The coverage probability goes to the nominal level for $\rho\in\{0,0.25,0.5\}$, while there is an under-coverage for $\rho=0.75$. Such phenomenon can also be found in some HAC robust tests, for example, \cite{andrews_heteroskedasticity_1991} and \cite{andrews_improved_1992}. \cite{kiefer_new_2005} has pointed out that HAC robust test tends to over-reject in finite samples if the traditional method of selecting bandwidth (where the bandwidth goes to infinity slower than the sample size) is used, especially when the serial correlation is strong. As the serial correlation gets larger, it is expected to pick a larger $l$ that captures more of the serial correlation. This suggests a data-driven method of selecting subsample size is needed for the quantile method as well.\\
\begin{table}[H]
        \centering
        \begin{tabular}{cccccc|cccccccc}
        \toprule
        $\rho$ & $N$ & $T$ & $b$ & $l$ & Quantile & $\rho$ & $N$ & $T$ & $b$ & $l$ & Quantile\\
        \midrule
        \midrule
        0 & 50 & 50 & 12 & 8 & 0.932 & 0.5 & 50 & 50 & 12 & 8 & 0.879\\
        0 & 100 & 100 & 15 & 10 & 0.967 & 0.5 & 100 & 100 & 15 & 10 & 0.924\\
        0 & 200 & 200 & 25 & 18 & 0.957 & 0.5 & 200 & 200 & 25 & 18 & 0.937\\
        0 & 400 & 400 & 49 & 35 & 0.947 &0.5 & 400 & 400 & 49 & 35 & 0.950\\
        \midrule
        0.25 & 50 & 50 & 12 & 8 & 0.936 & 0.75 & 50 & 50 & 12 & 8 & 0.711\\
        0.25 & 100 & 100 & 15 & 10 & 0.955 & 0.75 & 100 & 100 & 15 & 10 & 0.771\\
        0.25 & 200 & 200 & 25 & 18 & 0.951 & 0.75 & 200 & 200 & 25 & 18 & 0.837\\
        0.25 & 400 & 400 & 49 & 35 & 0.953 & 0.75 & 400 & 400 & 49 & 35 & 0.904\\
        \bottomrule
        \bottomrule
        \end{tabular}
        \caption{Coverage Probability for $\beta_{1}$ with a nominal probability of $95\%$ and 1000 Monte Carlo repetitions, where the confidence interval is constructed using the subsampling quantiles. $\rho$ is the AR coefficient controlling the level of serial correlation. $N$ and $T$ are the sample sizes. $b$ and $l$ are subsample sizes for individuals and time periods respectively.}
        \label{tb1}
\end{table}
\par
For the subsampling variance estimator, instead of subsampling the whole OLS estimators and then computing the subsample variance, we use the two estimators $\hat{V}_{N,T,b,l}^{2}$ and $\check{V}_{N,T,b,l}^{2}$, which are defined in Theorem \ref{thm3}, as well as their bias-corrected versions. The confidence intervals are constructed using the critical value of a standard normal distribution. The data-driven method discussed in Section \ref{sec3} is used to select the subsample size. The block length $l$ is set to be $\max\{4,l_{opt}\}$, where $l_{opt}$ is defined in \ref{lopt}), and $b$ is set to be equal to $\frac{N}{T}\cdot l$. The reason that we bound the subsample size from below is to leave room for the smaller subsample used in the bias-correction. $\Tilde{b}$ and $\Tilde{l}$ used for bias-correction are set to be $\lfloor\sqrt{b}\rfloor$ and $\lfloor\sqrt{l}\rfloor$ respectively. Table \ref{tb2} reports the coverage probabilities for $\beta_{1}$ of our subsampling variance estimators together with the CHS estimator (\cite{chiang_standard_2024}) and the Chen-Vogelsang estimator (CV; \cite{chen_fixed-b_2024}). The nominal probability is $95\%$, and we run 1000 Monte Carlo repetitions for each $\rho\in\{0,0.25,0.5,0.75\}$. The sample size $(N,T)$ vary from 50 to 200 for each choice of $\rho$. \\
\par
The results are as follow. First, both bias-corrected variance estimators, $\hat{V}_{N,T,b,l}^{2,BC}$ and $\check{V}_{N,T,b,l}^{2,BC}$, provide the correct nominal coverage probabilities with large sample size except for $\check{V}_{N,T,b,l}^{2,BC}$ under $\rho=0.75$. When the sample size is small, they have lower coverage due to the usage of a even smaller subsample. $\hat{V}_{N,T,b,l}^{2}$ and $\check{V}_{N,T,b,l}^{2}$ tend to be conservative when the serial correlation is weak. This is largely due to the subsample size selection. When the level of dependence is low, the algorithm usually returns a small optimal block length $l_{opt}$. However, by manually bounding $l$ from below, the effective block length $l$ tends to be longer and leads to a more conservative variance estimation. Meanwhile, CHS and CV are less robust to serial correlation than the subsampling variance estimators. These results are consistent with the asymptotic theory and show the effectiveness of the bias-correction under a moderate sample size.\\
\begin{table}[H]
    \centering
    \begin{tabular}{ccccccccc}
    \toprule
     $\rho$ & $N$ & $T$ & CHS & CV & $\hat{V}_{N,T,b,l}^{2}$ & $\hat{V}_{N,T,b,l}^{2,BC}$ & $\check{V}_{N,T,b,l}^{2}$ & $\check{V}_{N,T,b,l}^{2,BC}$\\
    \midrule
    \midrule
    0 & 50 & 50 & 0.948 & 0.952 & 0.949 & 0.885 & 0.945 & 0.875\\
    0 & 100 & 100 & 0.966 & 0.967 & 0.959 & 0.930 & 0.968 & 0.927\\
    0 & 200 & 200 & 0.972 & 0.972 & 0.976 & 0.944 & 0.968 & 0.941\\
    \midrule
    0.25 & 50 & 50 & 0.926 & 0.929 & 0.966 & 0.918 & 0.952 & 0.895\\
    0.25 & 100 & 100 & 0.940 & 0.942 & 0.969 & 0.945 & 0.963 & 0.930\\
    0.25 & 200 & 200 & 0.922 & 0.923 & 0.970 & 0.946 & 0.974 & 0.945\\
    \midrule
    0.5 & 50 & 50 & 0.880 & 0.893 & 0.961 & 0.923 & 0.936 & 0.896\\
    0.5 & 100 & 100 & 0.910 & 0.915 & 0.971 & 0.950 & 0.950 & 0.926\\
    0.5 & 200 & 200 & 0.905 & 0.909 & 0.960 & 0.948 & 0.954 & 0.945\\
    \midrule
    0.75 & 50 & 50 & 0.796 & 0.820 & 0.941 & 0.930 & 0.871 & 0.846\\
    0.75 & 100 & 100 & 0.848 & 0.857 & 0.926 & 0.933 & 0.906 & 0.907\\
    0.75 & 200 & 200 & 0.891 & 0.894 & 0.939 & 0.946 & 0.916 & 0.922\\
    \bottomrule
    \bottomrule
    \end{tabular}
    \caption{Coverage Probability for $\beta_{1}$ with a nominal probability of $95\%$ and 1000 Monte Carlo repetitions, where the confidence interval is constructed using the subsampling variance estimator. $\rho$ is the AR coefficient controlling the level of serial correlation. $N$ and $T$ are the sample sizes. $\hat{\sigma}_{N,T,b,l}^{2}$ is the subsampling variance estimator, while $\hat{\sigma}_{N,T,b,l}^{2,BC}$ is the bias-corrected version. The standard normal critical value is used.}
    \label{tb2}
\end{table}
\par
Furthermore, we compare our method with the CHS variance estimator and the wild bootstrap method designed by \cite{hounyo_reliable_2024}. We compare the 95\% confidence interval coverage probabilities for $\beta_{1}$ using $\hat{V}_{N,T,b,l}^{2}$, $\check{V}_{N,T,b,l}^{2}$, their bias-corection counterparts, two kinds of CHS variance estimator, and two kinds of wild bootstrap methods under 1000 Monte Carlo simulations. The sample size $N$ and $T$ are both set to be 100, and $\rho$ takes value between 0 to 0.9. The results are shown in Figure \ref{fig_compare}. The blue lines represents $\hat{V}_{N,T,b,l}^{2}$ and $\hat{V}_{N,T,b,l}^{2,BC}$, while $\check{V}_{N,T,b,l}^{2}$ and $\check{V}_{N,T,b,l}^{2,BC}$ are colored in aquamarine. The CHS variance estimator is proposed by \cite{chiang_standard_2024}, and CHS\_V is a variant of CHS variance estimator using a different weight function proposed by \cite{hounyo_reliable_2024}. MWCB\_equal and MWCB\_sym are two wild bootstrap methods using different P values. MWCB\_equal uses equal-tail bootstrap P values, while MWCB\_sym uses symmetric bootstrap P values. As we can see, the wild bootstrap methods remain robust and delivers the correct coverage level across different values of $\rho$. The two CHS variance estimators are relatively sensitive to the serial correlation, as they only have the correct coverage probabilities when $\rho<0.4$. Meanwhile, our theoretical subsampling variance estimators deliver the nominal coverage level and are comparable to the wild bootstrap method except when $\rho=0.9$. The practical subsampling variance estimators behave slightly worse but are still more robust to the serial correlation compared to two CHS estimators. Under the case when $\rho=0.9$ and the time effects are highly correlated, the subsampling variance estimators exhibit significant under-coverage. However, $\hat{V}_{N,T,b,l}^{2,BC}$ still has a coverage probability around 90\%, while its non-corrected counterpart only has about 85\%. This shows the effectiveness of the bias-correction. Overall, the subsampling method is comparable to the leading one of the existing methods in terms of the inference quality, except when the level of serial correlation is extremely high.\\
\begin{figure}[H]
    \centering
    \includegraphics[width=\textwidth]{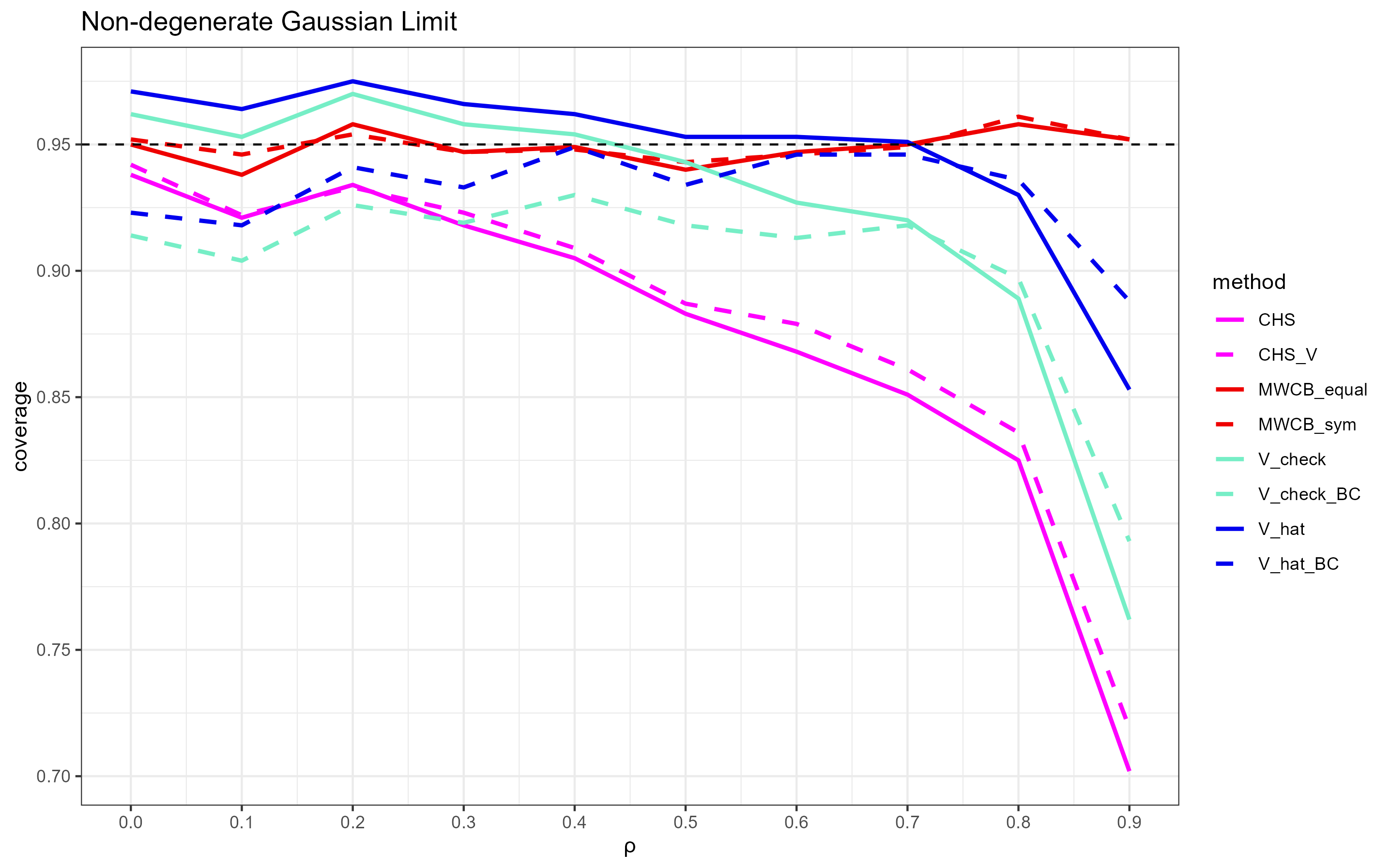}
    \caption{Coverage Probability for $\beta_{1}$ with a nominal probability of $95\%$ and 1000 Monte Carlo repetitions under various dependence level $\rho$. The two feasible subsampling variance estimators use the B\"uhlmann's iterative method to select subsample sizes. The two bootstrap methods use 500 bootstrap replications. The sample size is set to be $N=T=100$.}
    \label{fig_compare}
\end{figure}

\subsection{Degenerate Non-Gaussian Limit}
In the part, we consider a non-separable heterogeneity model (Example 1.7 in \cite{menzel_bootstrap_2021}).
\begin{equation*}
    X_{nt}=15\cdot\alpha_{n}\gamma_{t}+0.1\cdot\varepsilon_{nt}
\end{equation*}
$\alpha_{n},\gamma_{t}$, and $\varepsilon_{nt}$ are mutually independent. $\alpha_{n}$ and $\varepsilon_{nt}$ are two standard normal random variables, where $\gamma_{t}$ is an AR(1) process generated as
\begin{equation*}
    \gamma_{t+1}=\rho\cdot\gamma_{t}+v_{t}
\end{equation*}
where the innovation terms $v_{t}$ follows the distribution $N(0,1-\rho^{2})$ and is independent to $\gamma_{t+1}$. Since both $\alpha_{n}$ and $\gamma_{t}$ have an expectation of 0, then  $\frac{1}{\sqrt{NT}}\sum_{n=1}^{N}\sum_{t=1}^{T}X_{nt}\rightarrow_{d}15\cdot Z_{1}Z_{2}+0.1\cdot Z_{3}$ as described in Section \ref{sec2}. The coefficients 15 and 0.1 are chosen to amplify the non-Gaussianity. Our goal is to test the 95\% confidence interval coverage probability of the sample average $\frac{1}{\sqrt{NT}}\sum_{n=1}^{N}\sum_{t=1}^{T}X_{nt}$. For the subsampling quantile method, we again use a two-sided equal-tailed confidence interval. We vary $(N,T)$ from 50 to 200 for $\rho\in\{0,0.25,0.5,0.75\}$, and the subsample size $b$ and $l$ are chosen accordingly under each pair of $(N,T)$. As shown in Table \ref{tb3}, with the same subsample size choice across different $\rho$'s, the subsampling quantile method is slightly conservative when the serial correlation is weaker but remains overall valid. This little drawback can be mitigated by a data-driven bandwidth selection algorithm. However, as far as we know, there is no theoretical justified algorithm designed for such method.
\begin{table}[H]
        \centering
        \begin{tabular}{cccccc|cccccccc}
        \toprule
        $\rho$ & $N$ & $T$ & $b$ & $l$ & Quantile & $\rho$ & $N$ & $T$ & $b$ & $l$ & Quantile\\
        \midrule
        \midrule
        0 & 50 & 50 & 6 & 6 & 0.966 & 0.5 & 50 & 50 & 6 & 6 & 0.956\\
        0 & 100 & 100 & 8 & 8 & 0.962 & 0.5 & 100 & 100 & 8 & 8 & 0.942\\
        0 & 200 & 200 & 12 & 12 & 0.965 & 0.5 & 200 & 200 & 12 & 12 & 0.958\\
        \midrule
        0.25 & 50 & 50 & 6 & 6 & 0.968 & 0.75 & 50 & 50 & 6 & 6 & 0.944\\
        0.25 & 100 & 100 & 8 & 8 & 0.964 & 0.75 & 100 & 100 & 8 & 8 & 0.946\\
        0.25 & 200 & 200 & 12 & 12 & 0.969 & 0.75 & 200 & 200 & 12 & 12 & 0.943\\
        \bottomrule
        \bottomrule
        \end{tabular}
        \caption{Coverage Probability for $\Bar{X}_{NT}$ with a nominal probability of $95\%$ and 1000 Monte Carlo repetitions, where the confidence interval is constructed using the subsampling quantiles. $\rho$ is the AR coefficient controlling the level of serial correlation. $N$ and $T$ are the sample sizes. $b$ and $l$ are subsample sizes for individuals and time periods respectively.}
        \label{tb3}
\end{table}
Furthermore, we also compare the subsampling quantile method with the CHS variance estimator and the wild bootstrap method of \cite{hounyo_reliable_2024}. Under this degenerate non-Gaussian limit setting, CHS variance estimator and the wild bootstrap method are known to be theoretically invalid. The sample size $N$ and $T$ are both set to be 100, and $\rho$ takes value between 0 to 0.9. The results are shown in Figure \ref{fig_compare_degenerate}. With the non-Gaussian limit, the CLT based CHS variance estimator and the t-statistics based wild bootstrap method are too conservative. Especially, the wild bootstrap method consistently has a coverage probability around 99\%. Meanwhile the subsampling quantile method can generally provide the correct coverage probabilities except when the serial correlation is extremely strong. When $\rho<0.8$, the subsampling quantile method offers a coverage between 95\% and 96.5\%. This contrast highlights the advantage of our subsampling quantile method. It works under both Gaussian and non-Gaussian limits, as well as both non-degenerate and degenerate cases if we know whether it is degenerate or not.
\begin{figure}[H]
    \centering
    \includegraphics[width=\textwidth]{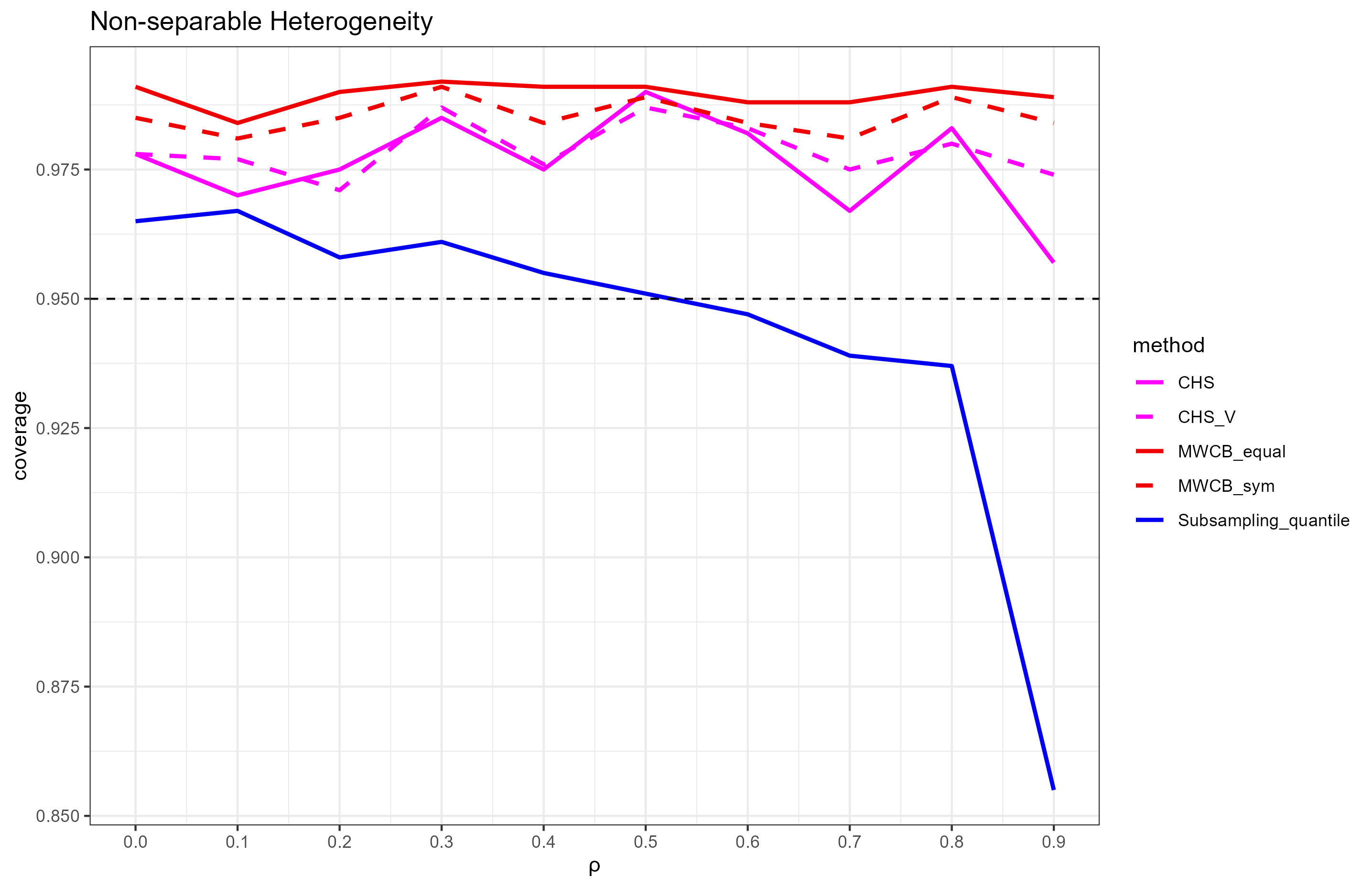}
    \caption{Coverage Probability for $\Bar{X}_{NT}$ with a nominal probability of $95\%$ and 1000 Monte Carlo repetitions under various dependence level $\rho$. The sample size is set to be $N=T=100$. The subsample sizes $b$ and $l$ are both set to be 8. The two bootstrap methods use 500 bootstrap replications.}
    \label{fig_compare_degenerate}
\end{figure}

Additional simulation results are in Appendix \ref{Append_sim}.

\section{Conclusion and Discussion}
In this paper, under a two-way clustering with serial correlation structure, we show the subsampling distribution converges to the limiting distribution in probability pointwise as long as the limit is well-defined. The subsampling distribution is defined to be the empirical distribution using all possible and proper subsets of the original sample. The pointwise convergence can be strengthened to uniform convergence if the limiting distribution is continuous. Hence, the quantiles of the subsampling distribution converge to their counterparts in the limiting distribution, suggesting the confidence intervals constructed from the subsampling distribution have the correct asymptotic coverage probabilities. Most importantly, such convergence does not depend on any specific form of the limiting distribution nor any specific convergence rate. We conduct simulation studies under different limiting distributions and convergence rates to confirm the validity of this quantile approach. Especially, the result suggests that the subsampling quantile method still provides the correct asymptotic coverage probability under a degenerate non-Gaussian limit, where all existing methods on two-way clustering with serial correlation fail. We further propose a consistent variance estimator of the limiting distribution, which is the variance of the subsampling distribution. Besides, a data-driven method to select subsample size for this variance estimator is introduced by connecting this subsampling variance estimator to the well-studied spectral density estimator, when the parameter of interest is the population average. We adopt the iterative method introduced by \cite{buhlmann_locally_1996}, and the optimal subsample sizes $b$ and $l$ are in the order of $N^{\frac{1}{3}}$ and $T^{\frac{1}{3}}$ respectively. Given this rate, we design a bias-corrected variance estimator which reduces the bias from $O(T^{-\frac{1}{3}})$ to $o(T^{-\frac{1}{3}})$. The bias-reduction is done by using another subsampling variance estimator with smaller subsample sizes to estimate the mean of the variance estimator. The simulation studies suggest our variance estimators are comparable, in terms of the inference quality, to the existing leading method, except when the serial correlation is extremely strong. When the serial correlation is extremely strong, our subsampling variance estimators rely on a large subsample size to deliver the desired coverage level. Nevertheless, when the serial correlation is moderate, our methods can deliver the correct coverage level under a reasonable sample size and are easy to implement.\\
\par
There are at least two improvements in the paper that could be made in the future. First, we have not found a proper subsample size selection algorithm for the quantile method nor the variance estimator with a general parameter of interest. Yet, one possible approach is to use the minimum volatility method. When the subsample size is too large, there are not enough variations between different $\hat{\theta}_{b,l,i,k}$ and they are all close to $\hat{\theta}_{NT}$. Then, the subsampling distribution will be too concentrated, and the subsampling confidence intervals tend to under-cover. When the subsample size is too small, the intervals can either over-cover or under-cover. The under-coverage comes from the possibility that $\frac{\tau_{bl}}{\tau_{NT}}$ being too small and resulting in a very tight confidence interval. When the subsample size is in the right range, we should expect the subsampling confidence intervals having coverage probabilities close to the nominal level. This says that for a certain range of the subsample size, the endpoints of confidence intervals should be smooth as a function of subsample size. Therefore, a heuristic method is to pick subsample size from a sequence such that the endpoints are least volatile at the pick. For the variance estimator, we pick the number that yields the least varying variance estimator. Unfortunately, there is no theorem that supports the optimality of such pick in subsampling. \cite{bickel_choice_2008} proposes a similar method that works for m-out-of-n bootstrap, which has a lot of similarity to subsampling, to select the bootstrap sample size $m$ and claims optimality. However, their theorem is built on the important fact that the bootstrap distribution converges anyway even if $\frac{m}{n}\rightarrow1$, while the convergence of subsampling distribution heavily relies on the subsample size in a smaller order of the original sample size. Therefore, a theorem of optimality is still needed under the subsampling framework. Secondly, both inference methods rely on a prior knowledge of the convergence rate, and it appears to be a more urging problem to the quantile method. \cite{politis_subsampling_1999} has shown that it is possible to apply subsampling methods to a problem with an unknown convergence rate (see Chapter 8). Yet, their method and result are for cross-sectional or time-series data, not panel data. Although a slight modification is probably enough to make their method work in the panel setting, details await to be shown.\\
\par
It might be possible to extend the theorems for subsampling under two-way clustering to a more sophisticated setting. For example, if the individuals are not i.i.d any more but spatially correlated, subsampling still has a chance to work. \cite{lahiri_prediction_1999} has shown that their proposed subsampling method provides accurate approximations to the sampling distributions of various functionals of the spatial cumulative distribution function predictor, especially the quantiles. Then, with an appropriate choice of the time indices, subsampling method could potentially work under two-way clustering with serial correlation and spatial dependence.
\newpage
\appendix
\textbf{\huge Appendix}
\section{Proof of Theorem \ref{thm1}}\label{AppendixA}
\begin{proof}
    (a): Define
    \[
    U_{N,T,b,l}(x):=\frac{1}{N_{b}\cdot q}\sum_{i=1}^{N_{b}}\sum_{k=1}^{q}\mathbbm{1}\{\tau_{bl}(\hat{\theta}_{b,l,i,k}-\theta)\leq x\}
    \]
    Fix $\delta>0$, we have $x-\delta\leq x-\tau_{bl}(\theta-\hat{\theta}_{NT})$ whenever $|\tau_{bl}(\theta-\hat{\theta}_{NT})|<\delta$. Define $E=\{\omega:|\tau_{bl}(\theta-\hat{\theta}_{NT})|<\delta\}$. Since $\tau_{bl}(\theta-\hat{\theta}_{NT})=\frac{\tau_{bl}}{\tau_{NT}}\cdot\tau_{NT}(\theta-\hat{\theta}_{NT})\rightarrow_{d}0$, then $\mathbb{P}(E)\rightarrow1$. Therefore, with a probability approaching to 1, for any $x\in\mathbb{R}$
    \begin{equation}\label{eq1}
        U_{N,T,b,l}(x-\delta)\leq L_{N,T,b,l}(x)\leq U_{N,T,b,l}(x+\delta)
    \end{equation}
    If $U_{N,T,b,l}(x)\rightarrow_{p}J(x)$ at every continuous point of $J$, and if $J$ is continuous at $x\pm\delta$, then (\ref{eq1}) implies
    \[
    J(x-\delta)-\delta\leq L_{N,T,b,l}(x)\leq J(x+\delta)+\delta
    \]
    with a probability going to 1. By continuity, for any $\varepsilon>0$, we can pick some $\delta>0$ such that
    \[
    J(x)-\varepsilon-\delta\leq L_{N,T,b,l}(x)\leq J(x)+\varepsilon+\delta
    \]
    with a probability going to 1. Since we can pick $\varepsilon$ and $\delta$ arbitrarily small, then $L_{N,T,b,l}(x)\rightarrow_{p}J(x)$ for any $x\in\mathbb{R}$ at which $J$ is continuous as desired.\\
    \\
    To show $U_{N,T,b,l}(x)\rightarrow_{p}J(x)$ at every continuous point of $J$, let $J$ be continuous at $x$ and observe
    \[
    \begin{split}
        \mathbb{E}[U_{N,T,b,l}(x)]&=\frac{1}{N_{b}\cdot q}\sum_{i=1}^{N_{b}}\sum_{k=1}^{q}\mathbb{E}[\mathbbm{1}\{\tau_{bl}(\hat{\theta}_{b,l,i,k}-\theta)\leq x\}]\\
        &=\mathbb{E}[\mathbbm{1}\{\tau_{bl}(\hat{\theta}_{b,l,i,1}-\theta)\leq x\}]\\
        &=\mathbb{P}(\tau_{bl}(\hat{\theta}_{b,l,i,1}-\theta)\leq x)\\
        &=J_{bl}(x)\rightarrow J(x)
    \end{split}    
    \]
    Then, it is suffice to show $\Var(U_{N,T,b,l}(x))\rightarrow0$. Let $E_{i,k}=\mathbbm{1}\{\tau_{bl}(\hat{\theta}_{b,l,i,k}-\theta)\leq x\}$. Then
    \begin{align}\label{eq2}
        \Var(U_{N,T,b,l}(x))&=\frac{1}{N_{b}^{2}q^{2}}\Var(\sum_{i=1}^{N_{b}}\sum_{k=1}^{q}\mathbbm{1}\{\tau_{bl}(\hat{\theta}_{b,l,i,k}-\theta)\leq x\})\notag\\
        &=\frac{1}{N_{b}^{2}q^{2}}[q\cdot\Var(\sum_{i=1}^{N_{b}}E_{i,1})+2\sum_{k=1}^{q-1}\sum_{h=1}^{q-k}\Cov(\sum_{i=1}^{N_{b}}E_{i,k},\sum_{i=1}^{N_{b}}E_{i,k+h})]\notag\\
        &=\frac{1}{N_{b}^{2}\cdot q}\Var(\sum_{i=1}^{N_{b}}E_{i,1})+\frac{2}{N_{b}^{2}\cdot q^{2}}\sum_{k=1}^{q-1}\sum_{h=1}^{q-k}\Cov(\sum_{i=1}^{N_{b}}E_{i,k},\sum_{i=1}^{N_{b}}E_{i,k+h})\notag\\
        &=\frac{1}{N_{b}^{2}\cdot q}[N_{b}\cdot\Var(E_{i,1})+2\sum_{i=1}^{N_{b}-1}\sum_{m=1}^{N_{b}-i}\Cov(E_{i,1},E_{i+m,1})]+\frac{2}{N_{b}^{2}\cdot q^{2}}\sum_{k=1}^{q-1}\sum_{h=1}^{q-k}\sum_{i=1}^{N_{b}}\sum_{j=1}^{N_{b}}\Cov(E_{i,k},E_{j,k+h})\notag\\
        &=\frac{1}{N_{b}\cdot q}\cdot\Var(E_{i,1})+\frac{2}{N_{b}^{2}\cdot q}\sum_{i=1}^{N_{b}-1}\sum_{m=1}^{N_{b}-i}\Cov(E_{i,1},E_{i+m,1})\notag\\
        &+\frac{2}{N_{b}^{2}\cdot q^{2}}\sum_{i=1}^{N_{b}}\sum_{j=1}^{N_{b}}\sum_{h=1}^{l-1}\sum_{k=1}^{q-h}\Cov(E_{i,k},E_{j,k+h})+\frac{2}{N_{b}^{2}\cdot q^{2}}\sum_{i=1}^{N_{b}}\sum_{j=1}^{N_{b}}\sum_{h=l}^{q-1}\sum_{k=1}^{q-h}\Cov(E_{i,k},E_{j,k+h})
    \end{align}
    The first term
    \begin{equation}\label{eq3}
        \frac{\Var(E_{i,1})}{N_{b}\cdot q}\rightarrow0
    \end{equation}
    as $\Var(E_{i,1})\in[0,1]$. For the second term, since $\Cov(E_{i,1},E_{i+m,1})$ is the covariance of two binary random variables, then $\Cov(E_{i,1},E_{i+m,1})\in[-1,1]$. Hence,
    \begin{equation}\label{eq4}
        |\frac{2}{N_{b}^{2}\cdot q}\sum_{i=1}^{N_{b}-1}\sum_{m=1}^{N_{b}-i}\Cov(E_{i,1},E_{i+m},1)|\leq\frac{2}{N_{b}^{2}\cdot q}\sum_{i=1}^{N_{b}-1}\sum_{m=1}^{N_{b}-i}1\leq\frac{2}{q}\rightarrow0
    \end{equation}
    For the third term, observe
    \begin{align}\label{eq5}
        |\frac{2}{N_{b}^{2}\cdot q^{2}}\sum_{i=1}^{N_{b}}\sum_{j=1}^{N_{b}}\sum_{h=1}^{l-1}\sum_{k=1}^{q-h}\Cov(E_{i,k},E_{j,k+h})|&=|\frac{2}{N_{b}^{2}\cdot q^{2}}\sum_{i=1}^{N_{b}}\sum_{j=1}^{N_{b}}\sum_{h=1}^{l-1}(q-h)\Cov(E_{i,k},E_{j,k+h})|\notag\\
        &\leq\frac{2}{N_{b}^{2}\cdot q^{2}}\sum_{i=1}^{N_{b}}\sum_{j=1}^{N_{b}}\sum_{h=1}^{l-1}(q-h)|\Cov(E_{i,k},E_{j,k+h})|\notag\\
        &\leq\frac{2}{N_{b}^{2}\cdot q^{2}}\sum_{i=1}^{N_{b}}\sum_{j=1}^{N_{b}}\sum_{h=1}^{l-1}(q-h)\notag\\
        &\leq\frac{2(l-1)}{q}\rightarrow0
    \end{align}
    For the fourth term, by Lemma \ref{le1}, $|\Cov(E_{i,k},E_{j,k+h})|\leq4\alpha_{\gamma}(h-l+1)$ whenever $h\geq l$. Therefore,
    \begin{align}\label{eq6}
        |\frac{2}{N_{b}^{2}\cdot q^{2}}\sum_{i=1}^{N_{b}}\sum_{j=1}^{N_{b}}\sum_{h=l}^{q-1}\sum_{k=1}^{q-h}\Cov(E_{i,k},E_{j,k+h})|&\leq\frac{8}{N_{b}^{2}\cdot q^{2}}\sum_{i=1}^{N_{b}}\sum_{j=1}^{N_{b}}\sum_{h=l}^{q-1}(q-h)\alpha_{\gamma}(h-l+1)\notag\\
        &\leq\frac{8}{q}\sum_{h=l}^{q-1}\alpha_{\gamma}(h-l+1)\notag\\
        &=\frac{8}{q}\sum_{m=1}^{q-l}\alpha_{\gamma}(m)\rightarrow0
    \end{align}
    Plugging (\ref{eq3}), (\ref{eq4}), (\ref{eq5}), and (\ref{eq6}) back to (\ref{eq2}) gives $\Var(U_{N,T,b,l}(x))\rightarrow0$, and this completes the proof of (a).\\
    \\
    (b): For the simplicity of notations, we may assume there exists an underlying index $m$ such that $m\rightarrow\infty$ is equivalent to $N,T\rightarrow\infty$. Let $\{F_{m}\}$ be a sequence of distribution functions, which can be either continuous or discontinuous, and $F$ be a continuous distribution function with $F_{m}(x)\rightarrow F(x)$ $\forall x\in\mathbb{R}$. For any $k\in\mathbb{N}$, define $x_{j,k}$ satisfying $F(x_{j,k})=\frac{j}{k}$, where $1\leq j\leq k-1$, and $x_{0,k}=-\infty$, $x_{k,k}=\infty$. Then $\forall x\in[x_{j,k},x_{j+1,k}]$
    \begin{align*}
        F_{m}(x)-F(x)&\leq F_{m}(x_{j+1,k})-F(x_{j,k})=F_{m}(x_{j+1,k})-F(x_{j+1,k})+\frac{1}{k}\\
        F_{m}(x)-F(x)&\geq F_{m}(x_{j,k})-F(x_{j+1,k})=F_{m}(x_{j,k})-F(x_{j,k})+\frac{1}{k}
    \end{align*}
    Hence,
    \[
    \sup_{x}|F_{m}(x)-F(x)|\leq\max_{0\leq j\leq k}2|F_{m}(x_{j,k})-F(x_{j,k})|+\frac{1}{k}
    \]
    For any $\varepsilon>0$, let $K\in\mathbb{N}$ satisfy $\frac{1}{K}<\frac{\varepsilon}{2}$, then
    \begin{align*}
        \mathbb{P}(\sup_{x}|F_{m}(x)-F(x)|>\varepsilon)&\leq\mathbb{P}(\max_{0\leq j\leq K}2|F_{m}(x_{j,K})-F(x_{j,K})|>\varepsilon-\frac{1}{K})\\
        &\leq\mathbb{P}(\max_{0\leq j\leq K}|F_{m}(x_{j,K})-F(x_{j,K})|>\frac{\varepsilon}{4})\\
        &\leq\sum_{j=0}^{K}\mathbb{P}(|F_{m}(x_{j,K})-F(x_{j,K})|>\frac{\varepsilon}{4})\rightarrow0
    \end{align*}
    Taking $F_{m}(x)=L_{N,T,b,l}(x)$ and $F(x)=J(x)$ completes the proof of (b).\\
    \\
    (c): By continuity, define $c_{1-\alpha}$ satisfying $J(c_{1-\alpha})=1-\alpha$. Since $L_{N,T,b,l}(c_{1-\alpha})\rightarrow_{p}J(c_{1-\alpha})$, then with a probability going to 1, for any $\varepsilon>0$,
    \begin{align*}
        L_{N,T,b,l}(c_{1-\alpha}-\varepsilon)&<1-\alpha<L_{N,T,b,l}(c_{1-\alpha}+\varepsilon)\\
        c_{1-\alpha}-\varepsilon&\leq c_{b,l}^{L}(1-\alpha)\leq c_{1-\alpha}+\varepsilon
    \end{align*}
    Hence, $J_{NT}(c_{1-\alpha}-\varepsilon)+op(1)\leq J_{NT}(c_{b,l}^{L}(1-\alpha))\leq J_{NT}(c_{1-\alpha}+\varepsilon)+op(1)$. Since $J_{NT}(c_{1-\alpha})\rightarrow J(c_{1-\alpha})$, then
    \begin{equation*}
        J(c_{1-\alpha}-\varepsilon)\leq\lim\inf J_{NT}(c_{b,l}^{L}(1-\alpha))\leq\lim\sup J_{NT}(c_{b,l}^{L}(1-\alpha))\leq J(c_{1-\alpha}+\varepsilon)
    \end{equation*}
    for any $\varepsilon>0$. By the continuity of $J$, $\lim J_{NT}(c_{b,l}^{L}(1-\alpha))=J(c_{1-\alpha})=1-\alpha$ as desired.
\end{proof}

\section{Proof of Theorem \ref{thm2}}\label{AppendixB}
\begin{proof}
    Since
    \begin{align*}
        \hat{\sigma}_{N,T,b,l}^{2}&=\frac{\tau_{bl}^{2}}{N_{b}\cdot q}\sum_{i=1}^{N_{b}}\sum_{k=1}^{q}(\hat{\theta}_{b,l,i,k}-\Bar{\theta}_{N,T,b,l})^{2}\\
        &=\frac{\tau_{bl}^{2}}{N_{b}\cdot q}\sum_{i=1}^{N_{b}}\sum_{k=1}^{q}(\hat{\theta}_{b,l,i,k}-\mathbb{E}[\hat{\theta}_{b,l,i,k}])^{2}-\tau_{bl}^{2}(\Bar{\theta}_{N,T,b,l}-\mathbb{E}[\hat{\theta}_{b,l,i,k}])^{2}\\
        &=\hat{S}_{N,T}-\hat{R}_{N,T}^{2}
    \end{align*}
    where
    \begin{align*}
        \hat{S}_{N,T}&=\frac{\tau_{bl}^{2}}{N_{b}\cdot q}\sum_{i=1}^{N_{b}}\sum_{k=1}^{q}(\hat{\theta}_{b,l,i,k}-\mathbb{E}[\hat{\theta}_{b,l,i,k}])^{2}\\
        \hat{R}_{N,T}^{2}&=\tau_{bl}^{2}(\Bar{\theta}_{N,T,b,l}-\mathbb{E}[\hat{\theta}_{b,l,i,k}])^{2}
    \end{align*}
    It is suffice to show $\hat{S}_{N,T}\rightarrow_{L^{2}}V$ and $\hat{R}_{N,T}^{2}\rightarrow_{L^{2}}0$.\\
    \\
    By strict stationarity and identical distribution of $\{\alpha_{n},\varepsilon_{nt}\}_{n=1}^{N}$ with a fixed $t$, we have $\mathbb{E}[\hat{S}_{N,T}]=\mathbb{E}[\tau_{bl}^{2}(\hat{\theta}_{b,l,i,k}-\mathbb{E}[\hat{\theta}_{b,l,i,k}])^{2}]\rightarrow V$. Next, we will show $\Var(\hat{S}_{N,T})\rightarrow0$. Let $\hat{T}_{b,l,i,k}^{2}=\tau_{bl}^{2}(\hat{\theta}_{b,l,i,k}-\mathbb{E}[\hat{\theta}_{b,l,i,k}])^{2}$.
    \begin{align}\label{eq7}
        \Var(\hat{S}_{N,T})&=\frac{1}{N_{b}^{2}q^{2}}\Var(\sum_{i=1}^{N_{b}}\sum_{k=1}^{q}\hat{T}_{b,l,i,k}^{2})\notag\\
        &\leq\frac{1}{N_{b}^{2}q^{2}}\sum_{k=1}^{q}\sum_{p=1}^{q}|\Cov(\sum_{i=1}^{N_{b}}\hat{T}_{b,l,i,k}^{2},\sum_{j=1}^{N_{b}}\hat{T}_{b,l,j,p}^{2})|\notag\\
        &=\frac{1}{N_{b}^{2}q^{2}}[\sum_{k=1}^{l}\sum_{p=1}^{q}|\Cov(\sum_{i=1}^{N_{b}}\hat{T}_{b,l,i,k}^{2},\sum_{j=1}^{N_{b}}\hat{T}_{b,l,j,p}^{2})|+\sum_{k=l+1}^{q-l}\sum_{p=1}^{k-l}|\Cov(\sum_{i=1}^{N_{b}}\hat{T}_{b,l,i,k}^{2},\sum_{j=1}^{N_{b}}\hat{T}_{b,l,j,p}^{2})|\notag\\
        &+\sum_{k=l+1}^{q-l}\sum_{p=k-l+1}^{k+l}|\Cov(\sum_{i=1}^{N_{b}}\hat{T}_{b,l,i,k}^{2},\sum_{j=1}^{N_{b}}\hat{T}_{b,l,j,p}^{2})|+\sum_{k=l+1}^{q-l}\sum_{p=k+l+1}^{q}|\Cov(\sum_{i=1}^{N_{b}}\hat{T}_{b,l,i,k}^{2},\sum_{j=1}^{N_{b}}\hat{T}_{b,l,j,p}^{2})|\notag\\
        &+\sum_{k=q-l+1}^{q}\sum_{p=1}^{q}|\Cov(\sum_{i=1}^{N_{b}}\hat{T}_{b,l,i,k}^{2},\sum_{j=1}^{N_{b}}\hat{T}_{b,l,j,p}^{2})|]\notag\\
        &=\frac{1}{N_{b}^{2}q^{2}}[A_{1}+A_{2}+A_{3}+A_{4}+A_{5}]
    \end{align}
    By uniform integrability, for any $i,j,k,p$, $|\Cov(\hat{T}_{b,l,i,k}^{2},\hat{T}_{b,l,j,p}^{2})|\leq M$ for some $M>0$. Then,
    \begin{align}\label{eq8}
        A_{1}&=\sum_{k=1}^{l}\sum_{p=1}^{q}|\Cov(\sum_{i=1}^{N_{b}}\hat{T}_{b,l,i,k}^{2},\sum_{j=1}^{N_{b}}\hat{T}_{b,l,j,p}^{2})|\notag\\
        &\leq\sum_{i=1}^{N_{b}}\sum_{j=1}^{N_{b}}\sum_{k=1}^{l}\sum_{p=1}^{q}|\Cov(\hat{T}_{b,l,i,k}^{2},\hat{T}_{b,l,j,p}^{2})|\notag\\
        &\leq N_{b}^{2}\cdot q\cdot l\cdot M
    \end{align}
    Similarly,
    \begin{align}\label{eq9}
        A_{3}&=\sum_{k=l+1}^{q-l}\sum_{p=k-l+1}^{k+l}|\Cov(\sum_{i=1}^{N_{b}}\hat{T}_{b,l,i,k}^{2},\sum_{j=1}^{N_{b}}\hat{T}_{b,l,j,p}^{2})|\notag\\
        &\leq\sum_{i=1}^{N_{b}}\sum_{j=1}^{N_{b}}\sum_{k=l+1}^{q-l}\sum_{p=k-l+1}^{k+l}|\Cov(\hat{T}_{b,l,i,k}^{2},\hat{T}_{b,l,j,p}^{2})|\notag\\
        &\leq N_{b}^{2}(q-2l)(2l-1)M\\
        A_{5}&=\sum_{k=q-l+1}^{q}\sum_{p=1}^{q}|\Cov(\sum_{i=1}^{N_{b}}\hat{T}_{b,l,i,k}^{2},\sum_{j=1}^{N_{b}}\hat{T}_{b,l,j,p}^{2})|\notag\\
        &\leq\sum_{i=1}^{N_{b}}\sum_{j=1}^{N_{b}}\sum_{k=q-l+1}^{q}\sum_{p=1}^{q}|\Cov(\hat{T}_{b,l,i,k}^{2},\hat{T}_{b,l,j,p}^{2})|\notag\\
        &\leq N_{b}^{2}\cdot q\cdot l\cdot M
    \end{align}
    By the uniform integrability, for any $\varepsilon>0$, $\exists A>0$ such that $\mathbb{E}[(\hat{T}_{b,l,i,k}^{2}\mathbbm{1}\{\hat{T}_{b,l,i,k}^{2}\geq A\})^{2}]<\frac{\varepsilon^{2}}{36M}$. For every fixed pair of $i,j$ and $k>p+l-1$, Lemma \ref{le3} then gives
    \begin{equation*}
        |\Cov(\hat{T}_{b,l,i,k}^{2},\hat{T}_{b,l,j,p}^{2})|\leq4A^{2}\alpha_{X}(k-p-l+1)+\varepsilon
    \end{equation*}
    Therefore,
    \begin{align}\label{eq10}
        A_{2}&=\sum_{k=l+1}^{q-l}\sum_{p=1}^{k-l}|\Cov(\sum_{i=1}^{N_{b}}\hat{T}_{b,l,i,k}^{2},\sum_{j=1}^{N_{b}}\hat{T}_{b,l,j,p}^{2})|\notag\\
        &\leq\sum_{i=1}^{N_{b}}\sum_{j=1}^{N_{b}}\sum_{k=l+1}^{q-l}\sum_{p=1}^{k-l}|\Cov(\hat{T}_{b,l,i,k}^{2},\hat{T}_{b,l,j,p}^{2})|\notag\\
        &\leq\sum_{i=1}^{N_{b}}\sum_{j=1}^{N_{b}}\sum_{k=l+1}^{q-l}\sum_{p=1}^{k-l}4A^{2}\alpha_{X}(k-p-l+1)+\varepsilon\notag\\
        &\leq N_{b}^{2}\sum_{h=l}^{q-l-1}\sum_{p=1}^{q-l-h}4A^{2}\alpha_{X}(h-l+1)+N_{b}^{2}(q-2l)^{2}\varepsilon\notag\\
        &\leq N_{b}^{2}(q-2l)\sum_{m=1}^{q-2l}4A^{2}\alpha_{X}(m)+N_{b}^{2}(q-2l)^{2}\varepsilon
    \end{align}
    Similarly,
    \begin{align}\label{eq11}
        A_{4}&=\sum_{k=l+1}^{q-l}\sum_{p=1}^{k+l+1}|\Cov(\sum_{i=1}^{N_{b}}\hat{T}_{b,l,i,k}^{2},\sum_{j=1}^{N_{b}}\hat{T}_{b,l,j,p}^{2})|\notag\\
        &\leq\sum_{i=1}^{N_{b}}\sum_{j=1}^{N_{b}}\sum_{k=l+1}^{q-l}\sum_{p=1}^{k+l+1}|\Cov(\hat{T}_{b,l,i,k}^{2},\hat{T}_{b,l,j,p}^{2})|\notag\\
        &\leq\sum_{i=1}^{N_{b}}\sum_{j=1}^{N_{b}}\sum_{k=l+1}^{q-l}\sum_{p=1}^{k+l+1}4A^{2}\alpha_{X}(k-p-l+1)+\varepsilon\notag\\
        &\leq N_{b}^{2}\sum_{h=l}^{q-l-1}\sum_{p=1}^{q-l-h}4A^{2}\alpha_{X}(h-l+1)+N_{b}^{2}(q-2l)^{2}\varepsilon\notag\\
        &\leq N_{b}^{2}(q-2l)\sum_{m=1}^{q-2l}4A^{2}\alpha_{X}(m)+N_{b}^{2}(q-2l)^{2}\varepsilon 
    \end{align}
    Plugging (\ref{eq8})-(\ref{eq11}) back to (\ref{eq7}) gives
    \begin{align*}
        \Var(\hat{S}_{N,T})&\leq(\frac{2l}{q}+\frac{(q-2l)(2l-1)}{q^{2}})M+\frac{8(q-2l)}{q^{2}}A^{2}\sum_{m=1}^{q-2l}\alpha_{X}(m)+\varepsilon
    \end{align*}
    By the fact that $\frac{l}{q}\rightarrow0$ and $\frac{1}{q}\sum_{m=1}^{q-2l}\alpha_{X}(m)\rightarrow0$, $\lim\sup\Var(\hat{S}_{N,T})\leq\varepsilon$ for any $\varepsilon>0$. Thus, $\Var(\hat{S}_{N,T})\rightarrow0$ and $\hat{S}_{N,T}\rightarrow_{p}V$. Since $\hat{S}_{N,T}$ is $L^{2}$-bounded and uniformly integrable, the convergence is also in $L^{2}$.\\
    \\
    Next, we will show $\hat{R}_{N,T}\rightarrow0$ in $L^{4}$. Recall $\hat{T}_{b,l,i,k}=\tau_{bl}(\hat{\theta}_{b,l,i,k}-\mathbb{E}[\hat{\theta}_{b,l,i,k}])$.
    \begin{align*}
        \Var(\hat{R}_{N,T})&=\frac{1}{N_{b}^{2}q^{2}}\Var(\sum_{i=1}^{N_{b}}\sum_{k=1}^{q}\hat{T}_{b,l,i,k})\\
        &\leq\frac{1}{N_{b}\cdot q}\Var(\hat{T}_{b,l,i,k})+\frac{1}{N_{b}^{2}q}\sum_{i=1}^{N_{b}}\sum_{j=1}^{N_{b}}|\Cov(\hat{T}_{b,l,i,k},\hat{T}_{b,l,j,k})|+\frac{1}{N_{b}^{2}q^{2}}\sum_{k=1}^{q}\sum_{p=1}^{q}|\Cov(\sum_{i=1}^{N_{b}}\hat{T}_{b,l,i,k},\sum_{j=1}^{N_{b}}\hat{T}_{b,l,j,p})|
    \end{align*}
    Since $\Var(\hat{T}_{b,l,i,k})=\mathbb{E}[\tau_{bl}^{2}(\hat{\theta}_{b,l,i,k}-\mathbb{E}[\hat{\theta}_{b,l,i,k}])^{2}]\rightarrow V$
    and $|\Cov(\hat{T}_{b,l,i,k},\hat{T}_{b,l,j,k})|\leq\Var^{\frac{1}{2}}(\hat{T}_{b,l,i,k})\Var^{\frac{1}{2}}(\hat{T}_{b,l,i,k})\rightarrow V$. The first two terms go to 0. The third term goes to 0 in a similar way as in (\ref{eq7}). Therefore, $\Var(\hat{R}_{N,T})\rightarrow0$. Thus, $\hat{R}_{N,T}\rightarrow_{p}0$ as $\mathbb{E}[\hat{R}_{N,T}]=0$. Since $\hat{\sigma}_{N,T,b,l}^{2}\geq 0$, $\hat{S}_{N,T}\geq\hat{R}_{N,T}^{2}$. Therefore, $\{\hat{R}_{N,T}^{4}\}$ is $L^{1}$ bounded and uniformly integrable. Thus, $\hat{R}_{N,T}\rightarrow0$ in $L^{4}$ and $\hat{\sigma}_{N,T,b,l}^{2}\rightarrow_{L^{2}}V$.
\end{proof}

\section{Proof of Theorem \ref{thm3}}\label{AppendixC}
\begin{proof}
    Let $\theta=\mathbb{E}[X_{nt}U_{nt}]$, $\phi=\mathbb{E}[X_{nt}X_{nt}']$, and
\begin{align*}
    \hat{\theta}_{NT}&=\frac{1}{NT}\sum_{n=1}^{N}\sum_{t=1}^{T}X_{nt}U_{nt}\\
    \hat{\psi}_{NT}&=\frac{1}{NT}\sum_{n=1}^{N}\sum_{t=1}^{T}X_{nt}\hat{U}_{nt}\\
    \hat{\phi}_{NT}&=\frac{1}{NT}\sum_{n=1}^{N}\sum_{t=1}^{T}X_{nt}X_{nt}'
\end{align*}
For some matrix $A$, define $A^{\otimes2}=AA'$. By Theorem 4 in \cite{chiang_standard_2024}, $\sqrt{N}(\hat{\beta}_{OLS}-\beta)\rightarrow_{d}N(0,V)$ for some $V=\mathbb{E}^{-1}[X_{nt}X_{nt}']\Sigma\mathbb{E}^{-1}[X_{nt}X_{nt}']$, where $\Sigma>0$ is defined by $\sqrt{N}(\hat{\theta}_{NT}-\theta)\rightarrow_{d}N(0,\Sigma)$. By Theorem \ref{thm2}, $\Sigma_{N,T,b,l}\rightarrow_{p}\Sigma$. By Theorem 1 in \cite{chiang_standard_2024}, $\hat{\phi}_{NT}\rightarrow_{p}\mathbb{E}[X_{nt}X_{nt}']$. Thus, $\hat{V}_{N,T,b,l}:=\hat{\phi}_{NT}^{-1}\hat{\Sigma}_{N,T,b,l}\hat{\phi}_{NT}^{-1}\rightarrow_{p}V$ and $\hat{t}:=\hat{V}_{N,T,b,l}^{-\frac{1}{2}}\sqrt{N}(\hat{\beta}_{OLS}-\beta)\rightarrow_{d}N(0,1)$.
\par
To show $\check{t}:=\check{V}_{N,T,b,l}^{-\frac{1}{2}}\sqrt{N}(\hat{\beta}-\beta)\rightarrow_{d}N(0,1)$. We first need to show $\check{\Sigma}_{N,T,b,l}\rightarrow_{p}\Sigma$. We can rewrite it as
\begin{align}
    \check{\Sigma}_{N,T,b,l}^{2}&=\frac{b}{N_{b}\cdot q}\sum_{i=1}^{N_{b}}\sum_{k=1}^{q}(\hat{\psi}_{b,l,i,k}-\Bar{\psi}_{N,T,b,l})^{\otimes2}\notag\\
    &=\frac{b}{N_{b}\cdot q}\sum_{i=1}^{N_{b}}\sum_{k=1}^{q}(\hat{\theta}_{b,l,i,k}-\Bar{\theta}_{N,T,b,l}+\hat{\phi}_{b,l,i,k}(\beta-\hat{\beta}_{OLS})-\Bar{\phi}_{N,T,b,l}(\beta-\hat{\beta}_{OLS}))^{\otimes2}\notag\\
    &=\hat{\Sigma}_{N,T,b,l}^{2}+\frac{b}{N_{b}\cdot q}\sum_{i=1}^{N_{b}}\sum_{k=1}^{q}[(\hat{\phi}_{b,l,i,k}-\Bar{\phi}_{N,T,b,l})(\beta-\hat{\beta}_{OLS})]^{\otimes2}\notag\\
    &+\frac{b}{N_{b}\cdot q}\sum_{i=1}^{N_{b}}\sum_{k=1}^{q}(\hat{\theta}_{b,l,i,k}-\Bar{\theta}_{N,T,b,l})(\hat{\phi}_{b,l,i,k}-\Bar{\phi}_{N,T,b,l})'(\beta-\hat{\beta}_{OLS})'\notag\\
    &+\frac{b}{N_{b}\cdot q}\sum_{i=1}^{N_{b}}\sum_{k=1}^{q}(\hat{\phi}_{b,l,i,k}-\Bar{\phi}_{N,T,b,l})(\beta-\hat{\beta}_{OLS})(\hat{\theta}_{b,l,i,k}-\Bar{\theta}_{N,T,b,l})'\notag\\
    &:=\hat{\Sigma}_{N,T,b,l}^{2}+A_{1}+A_{2}+A_{2}'
\end{align}
If we can show $A_{1}:=\frac{b}{N_{b}\cdot q}\sum_{i=1}^{N_{b}}\sum_{k=1}^{q}[(\hat{\phi}_{b,l,i,k}-\Bar{\phi}_{N,T,b,l})(\beta-\hat{\beta}_{OLS})]^{\otimes2}=o(1)$ and $A_{2}:=\frac{b}{N_{b}\cdot q}\sum_{i=1}^{N_{b}}\sum_{k=1}^{q}(\hat{\theta}_{b,l,i,k}-\Bar{\theta}_{N,T,b,l})(\hat{\phi}_{b,l,i,k}-\Bar{\phi}_{N,T,b,l})'(\beta-\hat{\beta}_{OLS})=o(1)$, then $\check{\Sigma}_{N,T,b,l}^{2}\rightarrow_{p}\Sigma$.\\
\\
To show $A_{1}=o(1)$, we first observe that $b(\hat{\beta}_{OLS}-\beta)^{\otimes2}=o(1)$. Furthermore, as $\mathbb{E}[||X_{nt}||^{8(r+\delta)}]<\infty$, $\Var(\hat{\phi}_{NT})\rightarrow0$. Applying Theorem \ref{thm2} on $\hat{\theta}_{NT}$ yields $\frac{1}{N_{b}\cdot q}\sum_{i=1}^{N_{b}}\sum_{k=1}^{q}(\hat{\phi}_{b,l,i,k}-\Bar{\phi}_{N,T,b,l})^{\otimes2}\rightarrow_{p}0$. Therefore, 
\begin{equation}\label{eq31}
    A_{1}:=\frac{b}{N_{b}\cdot q}\sum_{i=1}^{N_{b}}\sum_{k=1}^{q}[(\hat{\phi}_{b,l,i,k}-\Bar{\phi}_{N,T,b,l})(\beta-\hat{\beta}_{OLS})]^{\otimes2}=o(1)
\end{equation}
\newline\\
To show $A_{2}=o(1)$, as $b=o(N^{\frac{1}{2}})$, $b(\hat{\beta}_{OLS}-\beta)=o(1)$. Henceforth, as $\frac{1}{N_{b}\cdot q}\sum_{i=1}^{N_{b}}\sum_{k=1}^{q}(\hat{\theta}_{b,l,i,k}-\Bar{\theta}_{N,T,b,l})(\hat{\phi}_{b,l,i,k}-\Bar{\phi}_{N,T,b,l})'=\frac{1}{N_{b}\cdot q}\sum_{i=1}^{N_{b}}\sum_{k=1}^{q}\hat{\theta}_{b,l,i,k}\hat{\phi}_{b,l,i,k}'-\Bar{\theta}_{N,T,b,l}\Bar{\phi}_{N,T,b,l}'$, it is suffice to show
\begin{align}
    \frac{1}{N_{b}\cdot q}\sum_{i=1}^{N_{b}}\sum_{k=1}^{q}\hat{\theta}_{b,l,i,k}\hat{\phi}_{b,l,i,k}&\rightarrow_{p}0\label{eq32}\\
    \Bar{\theta}_{N,T,b,l}&\rightarrow_{p}\theta\label{eq33}\\
    \Bar{\phi}_{N,T,b,l}&\rightarrow_{p}\mathbb{E}[X_{nt}X_{nt}']\label{eq34}
\end{align}
We now prove (\ref{eq32}). Since $\mathbb{E}[U_{nt}|\mathbf{X}]=0$, the expectation of the LHS of (\ref{eq32}) is 0. It is left to prove $\Var(\frac{1}{N_{b}\cdot q}\sum_{i=1}^{N_{b}}\sum_{k=1}^{q}\hat{\theta}_{b,l,i,k}\hat{\phi}_{b,l,i,k})\rightarrow0$. As $\hat{\theta}_{b,l,i,k}\hat{\phi}_{b,l,i,k}$ is in fact a random matrix, for the sake of the essentially shorter notations and higher readability, we will only prove it under the scalar case. We can expand this variance as
\begin{align}\label{eq35}
    \Var(\frac{1}{N_{b}\cdot q}\sum_{i=1}^{N_{b}}\sum_{k=1}^{q}\hat{\theta}_{b,l,i,k}\hat{\phi}_{b,l,i,k})&=\frac{1}{N_{b}^{2}\cdot q^{2}}\sum_{i=1}^{N_{b}}\sum_{k=1}^{q}\sum_{j=1}^{N_{b}}\sum_{p=1}^{q}\Cov(\hat{\theta}_{b,l,i,k}\hat{\phi}_{b,l,i,k},\hat{\theta}_{b,l,j,p}\hat{\phi}_{b,l,j,p})\notag\\
    &=\frac{2}{N_{b}^{2}\cdot q^{2}}\sum_{i=1}^{N_{b}}\sum_{k=1}^{q}\sum_{j=1}^{N_{b}}\sum_{h=0}^{l-1}\Cov(\hat{\theta}_{b,l,i,k}\hat{\phi}_{b,l,i,k},\hat{\theta}_{b,l,j,k+h}\hat{\phi}_{b,l,j,k+h})\notag\\
    &+\frac{2}{N_{b}^{2}\cdot q^{2}}\sum_{i=1}^{N_{b}}\sum_{k=1}^{q}\sum_{j=1}^{N_{b}}\sum_{h=l}^{q-1}\Cov(\hat{\theta}_{b,l,i,k}\hat{\phi}_{b,l,i,k},\hat{\theta}_{b,l,j,k+h}\hat{\phi}_{b,l,j,k+h})
\end{align}
By Assumption \ref{a6} (iii) and H\"older's inequality, $\mathbb{E}[|\hat{\theta}_{b,l,i,k}\hat{\phi}_{b,l,i,k}|^{2r}]<M$ for some $M<\infty$. Observe that $\hat{\theta}_{b,l,i,k}\hat{\phi}_{b,l,i,k}$ is a measurable function of $\{(X_{nt},U_{nt})':n\in I_{i},k\leq t\leq k+l-1\}$. Then, by Lemma 5.1 in \cite{bradley_basic_2005} and Lemma 8.3.6 in \cite{durrett_probability_2019}, for any $i,j\in\{1,\dots,N_{b}\}$, $k\in\{1,\dots,q\}$ and $h\leq l$ we have
\begin{align}\label{eq36}
    |\Cov(\hat{\theta}_{b,l,i,k}\hat{\phi}_{b,l,i,k},\hat{\theta}_{b,l,j,k+h}\hat{\phi}_{b,l,j,k+h})|&\leq8\mathbb{E}^{\frac{1}{2r}}[|\hat{\theta}_{b,l,i,k}\hat{\phi}_{b,l,i,k}|^{2r}]\mathbb{E}^{\frac{1}{2r}}[|\hat{\theta}_{b,l,j,k+h}\hat{\phi}_{b,l,j,k+h}|^{2r}]\alpha^{\frac{r-1}{r}}(\mathcal{G}_{k}^{k+l-1},\mathcal{G}_{k+h}^{k+h+l-1})\notag\\
    &\leq8M^{\frac{1}{r}}\alpha_{\gamma}^{\frac{r-1}{r}}(h-l+1)
\end{align}
where $r$ satisfies Assumption \ref{a6}. Since $\gamma_{t}$ is strong mixing with size $\frac{2r}{r-1}$, then $\alpha_{\gamma}(k)$ decays faster than $k^{-\frac{2r}{r-1}}$ and $\alpha_{\gamma}^{\frac{r-1}{r}}(k)=O(k^{-2})$. As $\sum_{k=0}^{\infty}k^{-2}<\infty$, we have $\sum_{k=0}^{\infty}\alpha_{\gamma}^{\frac{r-1}{r}}(k)<\infty$. Plugging (\ref{eq36}) back to (\ref{eq35}) yields
\begin{align*}
    \Var(\frac{1}{N_{b}\cdot q}\sum_{i=1}^{N_{b}}\sum_{k=1}^{q}\hat{\theta}_{b,l,i,k}\hat{\phi}_{b,l,i,k})&\leq\frac{2}{N_{b}^{2}\cdot q^{2}}\sum_{i=1}^{N_{b}}\sum_{k=1}^{q}\sum_{j=1}^{N_{b}}\sum_{h=0}^{l-1}M+\frac{2}{N_{b}^{2}\cdot q^{2}}\sum_{i=1}^{N_{b}}\sum_{k=1}^{q}\sum_{j=1}^{N_{b}}\sum_{h=l}^{q-1}8M^{\frac{1}{r}}\alpha_{\gamma}^{\frac{r-1}{r}}(h-l+1)\\
    &=\frac{2l}{q}\cdot M+\frac{2}{q}\sum_{h=l}^{q-1}8M^{\frac{1}{r}}\alpha_{\gamma}^{\frac{r-1}{r}}(h-l+1)\\
    &=\frac{2l}{q}\cdot M+\frac{16M^{\frac{1}{r}}}{q}\sum_{k=1}^{q-l}\alpha_{\gamma}^{\frac{r-1}{r}}(k)\\
    &\rightarrow0
\end{align*}
Therefore, (\ref{eq32}) is proved. To prove (\ref{eq33}) and (\ref{eq34}), we can write
\begin{align}\label{eq37}
    \Bar{\theta}_{N,T,b,l}&=\frac{1}{N_{b}\cdot q}\sum_{i=1}^{N_{b}}\sum_{k=1}^{q}\frac{1}{bl}\sum_{n\in I_{i}}\sum_{t=k}^{k+l-1}X_{nt}U_{nt}\notag\\
    &=\frac{1}{N\cdot q}\sum_{n=1}^{N}\sum_{t=l}^{q}X_{nt}U_{nt}+\frac{1}{N\cdot q\cdot l}\sum_{n=1}^{N}\sum_{t=1}^{l-1}t\cdot X_{nt}U_{nt}+\frac{1}{N\cdot q\cdot l}\sum_{n=1}^{N}\sum_{t=q+1}^{T}(T-t+1)\cdot X_{nt}U_{nt}
\end{align}
By Theorem 1 in \cite{chiang_standard_2024}, the first term in (\ref{eq37}) converges to $\mathbb{E}[X_{nt}U_{nt}]=\theta$ in probability. For the second term, we have
\begin{equation*}
    \frac{1}{N\cdot q\cdot l}\sum_{n=1}^{N}\sum_{t=1}^{l-1}t\cdot X_{nt}U_{nt}\leq\frac{1}{N\cdot q}\sum_{n=1}^{N}\sum_{t=1}^{l-1}|X_{nt}U_{nt}|
\end{equation*}
Applying Theorem 1 in \cite{chiang_standard_2024} on $\frac{1}{N(l-1)}\sum_{n=1}^{N}\sum_{t=1}^{l-1}|X_{nt}U_{nt}|$ gives $\frac{1}{N(l-1)}\sum_{n=1}^{N}\sum_{t=1}^{l-1}|X_{nt}U_{nt}|\rightarrow_{p}\mathbb{E}[|X_{nt}U_{nt}|]$. Hence, the second term in (\ref{eq37}) converges to 0 in probability. Similarly, the third term in (\ref{eq37}) also converges to 0 in probability. Therefore $\Bar{\theta}_{N,T,b,l}\rightarrow_{p}\theta=0$, and (\ref{eq33}) is proved. By an eventually identical argument, $\Bar{\phi}_{N,T,b,l}\rightarrow_{p}\mathbb{E}[X_{nt}X_{nt}']$, and (\ref{eq34}) is proved. Therefore, $A_{2}=o(1)$. Thus, together with (\ref{eq31}), $\check{\Sigma}_{N,T,b,l}\rightarrow_{p}\Sigma$ and $\check{t}:=\check{V}_{N,T,b,l}^{-\frac{1}{2}}\sqrt{N}(\hat{\beta}-\beta)\rightarrow_{d}N(0,1)$.    
\end{proof}

\section{Technical Lemmas and Proofs}\label{Append_lemma}
\subsection{Lemma \ref{le1}}
\begin{lemma}\label{le1}
    Assume Assumption \ref{a1}-\ref{a3}. Let $g$ be a measurable function and bounded by $M<\infty$, then
    \[
    |\Cov(g(\hat{\theta}_{b,l,i,k}),g(\hat{\theta}_{b,l,j,k+h}))|\leq4M^{2}\alpha_{\gamma}(h-l+1)
    \]
    whenever $h\geq l$.
\end{lemma}
\begin{proof}
    Define $\mathcal{A}_{i}:=\sigma(\{\alpha_{n}:n\in I_{i}\})$, $\mathcal{G}_{k}^{p}:=\sigma(\{\gamma_{t}:k\leq t\leq p\})$, $\mathcal{E}_{k,i}^{p}:=\sigma(\{\varepsilon_{nt}:n\in I_{i},k\leq t\leq p\})$
    $\mathcal{F}_{k}^{p}(i):=\sigma(\{X_{nt}:n\in I_{i},k\leq t\leq p\})$. Observe that $g(\hat{\theta}_{b,l,i,k})$ is a measurable function of $\{X_{nt}:n\in I_{i},k\leq t\leq k+l-1\}$. Then, $\sigma(E_{i,k})\subset\mathcal{F}_{k}^{k+l-1}(i)\subset\mathcal{A}_{i}\vee\mathcal{G}_{k}^{k+l-1}\vee\mathcal{E}_{k,i}^{k+l-1}$. Therefore, by Lemma 8.3.6 in \cite{durrett_probability_2019} and Theorem 5.1 in \cite{bradley_basic_2005}
    \begin{align*}
        |\Cov(g(\hat{\theta}_{b,l,i,k}),g(\hat{\theta}_{b,l,j,k+h}))|&\leq4||g(\hat{\theta}_{b,l,i,k})||_{\infty}||g(\hat{\theta}_{b,l,j,k+h})||_{\infty}\alpha(\mathcal{F}_{k}^{k+l-1}(i),\mathcal{F}_{k+h}^{k+h+l-1}(j))\\
        &\leq4M^{2}[\alpha(\mathcal{A}_{i},\mathcal{A}_{j})+\alpha(\mathcal{G}_{k}^{k+l-1},\mathcal{G}_{k+h}^{k+h+l-1})+\alpha(\mathcal{E}_{k,i}^{k+l-1},\mathcal{E}_{k+h,j}^{k+h+l-1})]\\
        &=4M^{2}\alpha(\mathcal{G}_{k}^{k+l-1},\mathcal{G}_{k+h}^{k+h+l-1})\\
        &=4M^{2}\alpha_{\gamma}(h-l+1)
    \end{align*}
\end{proof}

\subsection{Lemma \ref{le3}}
\begin{lemma}\label{le3}
    Suppose $\{\tau_{NT}^{4}(\hat{\theta}_{NT}-\mathbb{E}[\hat{\theta}_{NT}])^{4}\}$ is uniformly integrable. Let $M<\infty$ satisfies $\mathbb{E}[\tau_{NT}^{4}(\hat{\theta}_{NT}-\mathbb{E}[\hat{\theta}_{NT}])^{4}]<M$ for any $N$ and $T$. Then under Assumption \ref{a1} and \ref{a3}, for every fixed pair of $i,j$, $k>p+l-1$, and any $\varepsilon>0$
    \begin{equation*}
        |\Cov(\hat{T}_{b,l,i,k}^{2},\hat{T}_{b,l,j,p}^{2})|\leq4A^{2}\alpha_{X}(k-p-l+1)+\varepsilon
    \end{equation*}
\end{lemma}

\begin{proof}
    The proof is based on Lemma 1 in \cite{carlstein_use_1986}. Note that $\hat{T}_{b,l,i,k}^{2}=\tau_{b,l}^{2}(\hat{\theta}_{b,l,i,k}-\mathbb{E}[\hat{\theta}_{b,l,i,k}])^{2}$ is $\mathcal{F}_{k}^{k+l-1}$-measurable and $\hat{T}_{b,l,j,p}^{2}$ is $\mathcal{F}_{p}^{p+l-1}$-measurable. Since $\mathbb{E}[\hat{T}_{b,l,i,k}^{4}]=\mathbb{E}[\hat{T}_{b,l,j,p}^{4}]<M$, then for every $A>0$
    \begin{align*}
        &|\Cov(\hat{T}_{b,l,i,k}^{2},\hat{T}_{b,l,j,p}^{2})|\\
        &\leq|\Cov(\hat{T}_{b,l,i,k}^{2}\mathbbm{1}\{\hat{T}_{b,l,i,k}^{2}< A\},\hat{T}_{b,l,j,p}^{2}\mathbbm{1}\{\hat{T}_{b,l,j,p}^{2}<A\})|\notag\\
        &+|\mathbb{E}[\hat{T}_{b,l,i,k}^{2}\mathbbm{1}\{\hat{T}_{b,l,i,k}^{2}<A\}\hat{T}_{b,l,j,p}^{2}\mathbbm{1}\{\hat{T}_{b,l,j,p}^{2}\geq A\}]|+|\mathbb{E}[\hat{T}_{b,l,i,k}^{2}\mathbbm{1}\{\hat{T}_{b,l,i,k}^{2}\geq A\}\hat{T}_{b,l,j,p}^{2}\mathbbm{1}\{\hat{T}_{b,l,j,p}^{2}<A\}]|\\
        &+|\mathbb{E}[\hat{T}_{b,l,i,k}^{2}\mathbbm{1}\{\hat{T}_{b,l,i,k}^{2}\geq A\}\hat{T}_{b,l,j,p}^{2}\mathbbm{1}\{\hat{T}_{b,l,j,p}^{2}\geq A\}]|+|\mathbb{E}[\hat{T}_{b,l,i,k}^{2}\mathbbm{1}\{\hat{T}_{b,l,i,k}^{2}<A\}]\mathbb{E}[\hat{T}_{b,l,j,p}^{2}\mathbbm{1}\{\hat{T}_{b,l,j,p}^{2}\geq A\}]|\\
        &+|\mathbb{E}[\hat{T}_{b,l,i,k}^{2}\mathbbm{1}\{\hat{T}_{b,l,i,k}^{2}\geq A\}]\mathbb{E}[\hat{T}_{b,l,j,p}^{2}\mathbbm{1}\{\hat{T}_{b,l,j,p}^{2}<A\}]|+|\mathbb{E}[\hat{T}_{b,l,i,k}^{2}\mathbbm{1}\{\hat{T}_{b,l,i,k}^{2}\geq A\}]\mathbb{E}[\hat{T}_{b,l,j,p}^{2}\mathbbm{1}\{\hat{T}_{b,l,j,p}^{2}\geq A\}]|
 \end{align*}
By Lemma \ref{le1}, $|\Cov(\hat{T}_{b,l,i,k}^{2}\mathbbm{1}\{\hat{T}_{b,l,i,k}^{2}< A\},\hat{T}_{b,l,j,p}^{2}\mathbbm{1}\{\hat{T}_{b,l,j,p}^{2}<A\})|\leq4A^{2}\alpha_{\gamma}(k-p-l+1)$. By Cauchy-Schwarz inequality
\begin{align*}
    |\mathbb{E}[\hat{T}_{b,l,i,k}^{2}\mathbbm{1}\{\hat{T}_{b,l,i,k}^{2}<A\}\hat{T}_{b,l,j,p}^{2}\mathbbm{1}\{\hat{T}_{b,l,j,p}^{2}\geq A\}]|&\leq|\mathbb{E}[(\hat{T}_{b,l,i,k}^{2}\mathbbm{1}\{\hat{T}_{b,l,i,k}^{2}<A\})^{2}]\mathbb{E}[(\hat{T}_{b,l,j,p}^{2}\mathbbm{1}\{\hat{T}_{b,l,j,p}^{2}\geq A\})^{2}]|^{\frac{1}{2}}\\
    &\leq M^{\frac{1}{2}}|\mathbb{E}[(\hat{T}_{b,l,j,p}^{2}\mathbbm{1}\{\hat{T}_{b,l,j,p}^{2}\geq A\})^{2}]|^{\frac{1}{2}}
\end{align*}
The rest terms can be bounded in the same way. Then, due to the strict stationarity and identical distribution, we have
\begin{align}\label{eq29}
    |\Cov(\hat{T}_{b,l,i,k}^{2},\hat{T}_{b,l,j,p}^{2})|&\leq4A^{2}\alpha_{\gamma}(k-p-l+1)+M^{\frac{1}{2}}(4\mathbb{E}^{\frac{1}{2}}[(\hat{T}_{b,l,i,k}^{2}\mathbbm{1}\{\hat{T}_{b,l,i,k}^{2}\geq A\})^{2}]+2\mathbb{E}^{\frac{1}{2}}[(\hat{T}_{b,l,j,p}^{2}\mathbbm{1}\{\hat{T}_{b,l,j,p}^{2}\geq A\})^{2}])\notag\\
    &=4A^{2}\alpha_{\gamma}(k-p-l+1)+6M^{\frac{1}{2}}\mathbb{E}^{\frac{1}{2}}[(\hat{T}_{b,l,i,k}^{2}\mathbbm{1}\{\hat{T}_{b,l,i,k}^{2}\geq A\})^{2}]{2}
\end{align}
By the uniform integrability of $\{\tau_{NT}^{4}(\hat{\theta}_{NT}-\mathbb{E}[\hat{\theta}_{NT}])^{4}\}$, for any $\varepsilon>0$, pick an $A>0$ such that
\begin{equation*}
    \mathbb{E}[(\hat{T}_{b,l,i,k}^{2}\mathbbm{1}\{\hat{T}_{b,l,i,k}^{2}\geq A\})^{2}]<\frac{\varepsilon^{2}}{36M}
\end{equation*}
Plugging this back to (\ref{eq29}) yields the desired result.
\end{proof}

\section{Other Technical Details}\label{Append_detail}
\subsection{The Derivation of $V$}\label{Append_detail1}
\begin{align}\label{eqV}
    V&=\lim\Var(\sqrt{N}\hat{\theta}_{NT})\notag\\
    &=\lim N\Var(\frac{1}{NT}\sum_{n=1}^{N}\sum_{t=1}^{T}a_{n}+b_{t}+e_{nt})\notag\\
    &=\lim N\Var(\frac{1}{N}\sum_{n=1}^{N}a_{n})+N\Var(\frac{1}{T}\sum_{t=1}^{T}b_{t})+N\Var(\frac{1}{NT}\sum_{n=1}^{N}\sum_{t=1}^{T}e_{nt})
\end{align}
The first term is
\begin{align}\label{eqVa}
    \Var(\frac{1}{N}\sum_{n=1}^{N}a_{n})&=\frac{\Var(a_{1})}{N}=\frac{V_{a}}{N}
\end{align}
The second term is
\begin{align}\label{eqVb}
    \Var(\frac{1}{T}\sum_{t=1}^{T}b_{t})&=\frac{\Var(b_{1})}{T}+\frac{2}{T^{2}}\sum_{t=1}^{T}\sum_{p=1}^{t-1}\Cov(b_{t},b_{p})\notag\\
    &=\frac{\Var(b_{1})}{T}+\frac{2}{T^{2}}\sum_{k=1}^{T-1}(T-k)\Cov(b_{1},b_{1+k})\notag\\
    &=\frac{V_{b}}{T}+\frac{2}{T}\sum_{k=1}^{T-1}(1-\frac{k}{T})R_{b}(k)
\end{align}
The third term is
\begin{align}\label{eqVe}
    \Var(\frac{1}{NT}\sum_{n=1}^{N}\sum_{t=1}^{T}e_{nt})&=\frac{1}{N^{2}T^{2}}\sum_{n=1}^{N}\sum_{m=1}^{N}\Cov(\sum_{t=1}^{T}e_{nt},\sum_{t=1}^{T}e_{mt})\notag\\
    &=\frac{1}{N^{2}T^{2}}\sum_{n=1}^{N}\Var(\sum_{t=1}^{T}e_{nt})\notag\\
    &=\frac{1}{NT^{2}}[T\cdot\Var(e_{11})+2\sum_{k=1}^{T-1}(T-k)\Cov(e_{11},e_{1,1+k})]\notag\\
    &=\frac{V_{e}}{NT}+\frac{2}{NT}\sum_{k=1}^{T-1}(1-\frac{k}{T})R_{e}(k)
\end{align}
Plugging (\ref{eqVa})-(\ref{eqVe}) back to (\ref{eqV}) yields
\begin{align*}
    V&=\lim V_{a}+\frac{N}{T}V_{b}+2\frac{N}{T}\sum_{k=1}^{T-1}(1-\frac{k}{T})R_{b}(k)+\frac{V_{e}}{T}+\frac{2}{T}\sum_{k=1}^{T-1}(1-\frac{k}{T})R_{e}(k)\\
    &=V_{a}+cV_{b}+2c\sum_{k=1}^{\infty}R_{b}(k)
\end{align*}
To see $V<\infty$, first observe $V_{a}=\mathbb{E}[a_{n}^{2}]=\mathbb{E}[\mathbb{E}^{2}[X_{nt}|\alpha_{n}]]\leq\mathbb{E}[X_{nt}^{2}]<\infty$. And for every $k\in\mathbb{Z}_{\geq0}$, by Lemma 8.3.6 in \cite{durrett_probability_2019}
\begin{align*}
    |R_{b}(k)|\leq8\mathbb{E}^{\frac{1}{2r}}[|X_{nt}|^{4r}]\alpha_{\gamma}^{\frac{2r-1}{2r}}(k)
\end{align*}
Since $\gamma_{t}$ is strong mixing with size $\frac{2r}{r-1}$ (Assumption \ref{a4} (ii)), then $\alpha_{\gamma}(k)$ decays faster than $k^{-\frac{2r}{r-1}}$ and $\alpha_{\gamma}^{\frac{2r-1}{2r}}(k)=O(k^{-\frac{2r-1}{r-1}})$. As $\sum_{k=0}^{\infty}k^{-\frac{2r-1}{r-1}}<\infty$, $\sum_{k=0}^{\infty}\alpha_{\gamma}^{\frac{2r-1}{2r}}(k)<\infty$. Hence, by $\mathbb{E}[|X_{nt}|^{4r}]<\infty$, $|\sum_{k=0}^{\infty}R_{b}(k)|<\infty$. Thus, $V<\infty$.
\subsection{The Finiteness of $\sum_{k=-\infty}^{\infty}|kR_{b}(k)|$}\label{Append_detail2}
For every $k\in\mathbb{Z}_{\geq0}$, $|R_{b}(k)|\leq8\mathbb{E}^{\frac{1}{2r}}[|X_{nt}|^{4r}]\alpha_{\gamma}^{\frac{2r}{2r-1}}(k)$, then $|kR_{b}(k)|\leq8\mathbb{E}^{\frac{1}{2r}}[|X_{nt}|^{4r}]k\alpha_{\gamma}^{\frac{2r}{2r-1}}(k)$. Since $\gamma_{t}$ is strong mixing with size $\frac{2r}{r-1}$, then $\alpha_{\gamma}^{\frac{2r-1}{2r}}(k)=O(k^{-\frac{2r-1}{r-1}})$ and $k\alpha_{\gamma}^{\frac{2r}{2r-1}}(k)=O(k^{-\frac{r}{r-1}})$. Then, $\sum_{k=1}^{\infty}k^{-\frac{r}{r-1}}<\infty$ implies $\sum_{k=-\infty}^{\infty}|k|\alpha_{\gamma}^{\frac{2r}{2r-1}}(|k|)<\infty$. Thus, $\sum_{k=-\infty}^{\infty}|kR_{b}(k)|<\infty$.
\subsection{The Derivation of $Bias(\hat{\sigma}_{N,T,b,l}^{2},V)$}\label{Append_detail3}
\begin{align}\label{eqsighatmean}
    \mathbb{E}[\hat{\sigma}_{N,T,b,l}^{2}]&=\mathbb{E}[\frac{b}{N_{b}\cdot q}\sum_{i=1}^{N_{b}}\sum_{k=1}^{q}(\hat{\theta}_{b,l,i,k}-\Bar{\theta}_{N,T,b,l})^{2}]\notag\\
    &=\mathbb{E}[\frac{b}{N_{b}\cdot q}\sum_{i=1}^{N_{b}}\sum_{k=1}^{q}\hat{\theta}_{b,l,i,k}^{2}-b\cdot\Bar{\theta}_{N,T,b,l}^{2}]\notag\\
    &=b\cdot\Var(\hat{\theta}_{b,l,i,k})-b\cdot\mathbb{E}[\Bar{\theta}_{N,T,b,l}^{2}]
\end{align}
The first part is equal to
\begin{align}\label{eqsighatmean_1}
    b\cdot\Var(\hat{\theta}_{b,l,i,k})&=V_{a}+\frac{b}{l}\sum_{k=-l+1}^{l-1}(1-\frac{|k|}{l})R_{b}(k)+\frac{1}{l}\sum_{k=-l+1}^{l-1}(1-\frac{|k|}{l})R_{e}(k)
\end{align}
while the second part can be decomposed as
\begin{align}\label{eqsighatmean_2}
    b\cdot\mathbb{E}[\Bar{\theta}_{N,T,b,l}^{2}]&=b\cdot\mathbb{E}[(\frac{1}{N_{b}\cdot q}\sum_{i=1}^{N_{b}}\sum_{k=1}^{q}\hat{\theta}_{b,l,i,k})^{2}]\notag\\
    &=b\cdot\mathbb{E}[(\frac{1}{N_{b}}\sum_{i=1}^{N_{b}}\frac{1}{b}\sum_{n\in I_{i}}a_{n}+\frac{1}{q}\sum_{k=1}^{q}\frac{1}{l}\sum_{t=k}^{k+l-1}b_{t}+\frac{1}{N_{b}\cdot q}\sum_{i=1}^{N_{b}}\sum_{k=1}^{q}\frac{1}{bl}\sum_{n\in I_{i}}\sum_{t=k}^{k+l-1}e_{nt})^{2}]\notag\\
    &=b\cdot\mathbb{E}[(\frac{1}{N}\sum_{n=1}^{N}a_{n}+\frac{1}{q\cdot l}(\sum_{t=1}^{l-1}t\cdot b_{t}+\sum_{t=q+1}^{T}(T-t+1)b_{t}+\sum_{t=l}^{q}l\cdot b_{t})\notag\\
    &+\frac{1}{N\cdot q\cdot l}\sum_{n=1}^{N}(\sum_{t=1}^{l-1}t\cdot e_{nt}+\sum_{t=q+1}^{T}(T-t+1)e_{nt}+\sum_{t=l}^{q}l\cdot e_{nt}))^{2}]\notag\\
    &=b\cdot\mathbb{E}[(\frac{1}{N}\sum_{n=1}^{N}a_{n})^{2}]+b\cdot\mathbb{E}[\frac{1}{q^{2}l^{2}}(\sum_{t=1}^{T}l\cdot b_{t}-\underbrace{\sum_{t=1}^{l}(l-t)b_{t}}_{S_{1}}-\underbrace{\sum_{t=q}^{T}(T-t+1-l)b_{t}}_{S_{2}})^{2}]\notag\\
    &+b\cdot\mathbb{E}[\frac{1}{N^{2}q^{2}l^{2}}(\sum_{n=1}^{N}(\sum_{t=1}^{l-1}t\cdot e_{nt}+\sum_{t=q+1}^{T}(T-t+1)e_{nt}+\sum_{t=l}^{q}l\cdot e_{nt}))^{2}]\notag\\
    &=\frac{b}{N}V_{a}+\frac{bT}{q^{2}}V_{b}+2\frac{bT}{q^{2}}\sum_{k=1}^{T-1}(1-\frac{k}{T})R_{b}(k)+\frac{b}{q^{2}l^{2}}\mathbb{E}[S_{1}^{2}+S_{2}^{2}+2S_{1}S_{2}]-\frac{2b}{q^{2}l^{2}}\mathbb{E}[(\sum_{t=1}^{T}l\cdot b_{t})(S_{1}+S_{2})]\notag\\
    &+b\cdot\mathbb{E}[\frac{1}{N^{2}q^{2}l^{2}}(\sum_{n=1}^{N}(\sum_{t=1}^{l-1}t\cdot e_{nt}+\sum_{t=q+1}^{T}(T-t+1)e_{nt}+\sum_{t=l}^{q}l\cdot e_{nt}))^{2}]
\end{align}
For the term $\frac{b}{q^{2}l^{2}}\mathbb{E}[S_{1}^{2}+S_{2}^{2}+2S_{1}S_{2}]$, observe
\begin{align*}
    \frac{b}{q^{2}l^{2}}\mathbb{E}[S_{1}^{2}]&=\frac{b}{q^{2}l^{2}}\mathbb{E}[\sum_{t=1}^{l}\sum_{p=1}^{l}(l-t)(l-p)b_{t}b_{p}]\\
    &\leq\frac{b(l-1)^{2}}{q^{2}l^{2}}\mathbb{E}[\sum_{t=1}^{l}\sum_{p=1}^{l}b_{t}b_{p}]\\
    &\leq\frac{b}{q^{2}}\sum_{k=-l}^{k=l}(l-|k|)R_{b}(k)\\
    &=O(\frac{bl}{N^{2}})
\end{align*}
Similarly, $\frac{b}{q^{2}l^{2}}\mathbb{E}[S_{2}^{2}],\frac{b}{q^{2}l^{2}}\mathbb{E}[2S_{1}S_{2}]=O(\frac{bl}{N^{2}})$. Hence,
\begin{equation}\label{eqS1S2}
    \frac{b}{q^{2}l^{2}}\mathbb{E}[S_{1}^{2}+S_{2}^{2}+2S_{1}S_{2}]=O(\frac{bl}{N^{2}})
\end{equation}
For the term $\frac{2b}{q^{2}l^{2}}\mathbb{E}[(\sum_{t=1}^{T}l\cdot b_{t})(S_{1}+S_{2})]$, observe
\begin{align*}
    |\frac{2b}{q^{2}l^{2}}\mathbb{E}[(\sum_{t=1}^{T}l\cdot b_{t})S_{1}]|&=|\frac{2b}{q^{2}l^{2}}\mathbb{E}[(\sum_{t=1}^{T}l\cdot b_{t})(\sum_{t=1}^{l}(l-t)b_{t})]|\\
    &\leq\frac{2b}{q^{2}}|\mathbb{E}[\sum_{t=1}^{T}\sum_{p=1}^{l}b_{t}b_{p}]|\\
    &\leq\frac{2b}{q^{2}}(|\sum_{k=-l}^{k=l}(l-|k|)R_{b}(k)|+|\mathbb{E}[\sum_{t=l+1}^{T}\sum_{p=1}^{l}b_{t}b_{p}]|)\\
    &=\frac{2b}{q^{2}}(|\sum_{k=-l}^{k=l}(l-|k|)R_{b}(k)|+|\sum_{k=1}^{T-1}h(k)R_{b}(k)|)\\
    &=O(\frac{b}{N})
\end{align*}
where $h(k)=\begin{cases}
        l-|k-l|, & k<l\\
        l, & l\leq k\leq T-l\\
        l-|k-T+l|, & k>T-l
    \end{cases}$.\\
Hence,
\begin{equation}
    \frac{2b}{q^{2}l^{2}}\mathbb{E}[(\sum_{t=1}^{T}l\cdot b_{t})(S_{1}+S_{2})]=O(\frac{b}{N})
\end{equation}
For the term $b\cdot\mathbb{E}[\frac{1}{N^{2}q^{2}l^{2}}(\sum_{n=1}^{N}(\sum_{t=1}^{l-1}t\cdot e_{nt}+\sum_{t=q+1}^{T}(T-t+1)e_{nt}+\sum_{t=l}^{q}l\cdot e_{nt}))^{2}]$, $\mathbb{E}[e_{nt}e_{mp}]=0$ for any $n\neq m$ and any $t,p\in\mathbb{N}$. Then, 
\begin{align}\label{eqeterms}
    &b\cdot\mathbb{E}[\frac{1}{N^{2}q^{2}l^{2}}(\sum_{n=1}^{N}(\sum_{t=1}^{l-1}t\cdot e_{nt}+\sum_{t=q+1}^{T}(T-t+1)e_{nt}+\sum_{t=l}^{q}l\cdot e_{nt}))^{2}]\notag\\
    &=\frac{b}{N^{2}q^{2}l^{2}}\cdot\sum_{n=1}^{N}\mathbb{E}[(\sum_{t=1}^{l-1}t\cdot e_{nt}+\sum_{t=q+1}^{T}(T-t+1)e_{nt}+\sum_{t=l}^{q}l\cdot e_{nt})^{2}]\notag\\
    &\leq\frac{b}{Nq^{2}}\mathbb{E}[(\sum_{t=1}^{T}e_{nt})^{2}]\notag\\
    &=\frac{b}{Nq^{2}}\sum_{k=-T+1}^{T-1}(T-|k|)R_{e}(k)\notag\\
    &=o(\frac{b}{N})
\end{align}
Plugging (\ref{eqS1S2})-(\ref{eqeterms}) back to (\ref{eqsighatmean_2}) yields
\begin{equation}\label{eqsighatmean_22}
    b\cdot\mathbb{E}[\Bar{\theta}_{N,T,b,l}^{2}]=O(\frac{b}{N})+O(\frac{bl}{N^{2}})\\
\end{equation}
Therefore, plugging (\ref{eqsighatmean_1}) and (\ref{eqsighatmean_22}) back to (\ref{eqsighatmean}) gives
\begin{equation*}
   \mathbb{E}[\hat{\sigma}_{N,T,b,l}^{2}]=V_{a}+\frac{b}{l}\sum_{k=-l+1}^{l-1}(1-\frac{|k|}{l})R_{b}(k)+\frac{1}{l}\sum_{k=-l+1}^{l-1}(1-\frac{|k|}{l})R_{e}(k)+O(\frac{b}{N})+O(\frac{bl}{N^{2}})\\
\end{equation*}
Thus, we have the bias of $\hat{\sigma}_{N,T,b,l}^{2}$ as
\begin{align*}
    Bias(\hat{\sigma}_{N,T,b,l}^{2},V)&=\mathbb{E}[\hat{\sigma}_{N,T,b,l}^{2}]-V\notag\\
    &=-\frac{b}{l^{2}}\sum_{k=-\infty}^{\infty}|k|R_{b}(k)+(\frac{b}{l}-c)\sum_{k=-\infty}^{\infty}R_{b}(k)+2\frac{b}{l}\sum_{k=l}^{\infty}(\frac{k}{l}-1)R_{b}(k)\notag\\
    &+\frac{1}{l}\sum_{k=-l+1}^{l-1}(1-\frac{|k|}{l})R_{e}(k)+O(\frac{b}{N})+O(\frac{bl}{N^{2}})
\end{align*}
\subsection{The Derivation of $Bias(\hat{\sigma}_{N,T,b,l}^{2,BC},V)$}
\begin{align*}
    Bias(\hat{\sigma}_{N,T,b,l}^{2,BC},V)&=\mathbb{E}[\hat{\sigma}_{N,T,b,l}^{2,BC}]-V\\
    &=\mathbb{E}[\hat{\sigma}_{N,T,b,l}^{2}]-V-D\cdot(\mathbb{E}[\hat{\sigma}_{N,T,\Tilde{b},\Tilde{l}}^{2}]-V-\mathbb{E}[\hat{\sigma}_{N,T,b,l}^{2}]+V)\\
    &=Bias(\hat{\sigma}_{N,T,b,l}^{2},V)-D\cdot(Bias(\hat{\sigma}_{N,T,\Tilde{b},\Tilde{l}}^{2},V)-Bias(\hat{\sigma}_{N,T,b,l}^{2},V))\\
    &=-\frac{1}{l}\sum_{k=-\infty}^{\infty}c|k|R_{b}(k)+\frac{1}{l}\sum_{k=-l+1}^{l-1}(1-\frac{|k|}{l})R_{e}(k)+o(\frac{1}{l})\\
    &-\frac{\Tilde{l}}{l-\Tilde{l}}\cdot(-\frac{1}{\Tilde{l}}\sum_{k=-\infty}^{\infty}c|k|R_{b}(k)+\frac{1}{\Tilde{l}}\sum_{k=-\Tilde{l}+1}^{\Tilde{l}-1}(1-\frac{|k|}{\Tilde{l}})R_{e}(k)+\frac{1}{l}\sum_{k=-\infty}^{\infty}c|k|R_{b}(k)-\frac{1}{l}\sum_{k=-l+1}^{l-1}(1-\frac{|k|}{l})R_{e}(k))\\
    &=[-\frac{1}{l}-\frac{\Tilde{l}}{l-\Tilde{l}}\cdot(-\frac{1}{\Tilde{l}}+\frac{1}{l})]\cdot[\sum_{k=-\infty}^{\infty}c|k|R_{b}(k)-\sum_{k=-\Tilde{l}+1}^{\Tilde{l}-1}R_{e}(k)]+\frac{2}{l}(1+\frac{\Tilde{l}}{l-\Tilde{l}})\sum_{k=\Tilde{l}}^{l-1}R_{e}(k)\\
    &-\frac{1}{l^{2}}(1+\frac{\Tilde{l}}{l-\Tilde{l}})\sum_{k=-l+1}^{l-1}|k|R_{e}(k)+\frac{1}{l-\Tilde{l}}\sum_{k=-\Tilde{l}+1}^{\Tilde{l}-1}\frac{|k|}{\Tilde{l}}R_{e}(k)+o(\frac{1}{l})\\
    &=\frac{2}{l-\Tilde{l}}\sum_{k=\Tilde{l}}^{l-1}R_{e}(k)-\frac{1}{l-\Tilde{l}}\sum_{k=-l+1}^{l-1}\frac{|k|}{l}R_{e}(k)+\frac{1}{l-\Tilde{l}}\sum_{k=-\Tilde{l}+1}^{\Tilde{l}-1}\frac{|k|}{\Tilde{l}}R_{e}(k)+o(\frac{1}{l})\\
\end{align*}
\subsection{The Equality $\mathbb{E}[X_{nt}X_{m,t+k}]=\mathbb{E}[b_{t}b_{t+k}]$}
\begin{align*}
    \mathbb{E}[X_{nt}X_{m,t+k}]&=\mathbb{E}[(a_{n}+b_{t}+e_{nt})(a_{m}+b_{t+k}+e_{m,t+k})]\\
    &=\mathbb{E}[a_{n}a_{m}+a_{n}b_{t+k}+a_{n}e_{m,t+k}+b_{t}a_{m}+b_{t}b_{t+k}+b_{t}e_{m,t+k}+e_{nt}a_{m}+e_{nt}b_{t+k}+e_{nt}e_{m,t+k}]
\end{align*}
By independence, $\mathbb{E}[a_{n}a_{m}]=\mathbb{E}[a_{n}b_{t+k}]=\mathbb{E}[b_{t}a_{m}]=0$. By the fact that $\{a_{n}\}$, $\{b_{t}\}$, and $\{e_{nt}\}$ are mutually uncorrelated, $\mathbb{E}[e_{nt}(a_{m}+b_{t+k})]=\mathbb{E}[e_{m,t+k}(a_{n}+b_{t})]=0$. By the conditional independence between $e_{nt}$ and $e_{mp}$ conditional on $(\gamma_{t},\gamma_{p})$, we have
\begin{align*}
    \mathbb{E}[e_{nt}e_{m,t+k}]&=\mathbb{E}[\mathbb{E}[e_{nt}e_{m,t+k}|\gamma_{t},\gamma_{t+k}]]\\
    &=\mathbb{E}[\mathbb{E}[e_{nt}|\gamma_{t},\gamma_{t+k}]\mathbb{E}[e_{m,t+k}|\gamma_{t},\gamma_{t+k}]]\\
    &=\mathbb{E}[\mathbb{E}[e_{nt}|\gamma_{t}]\mathbb{E}[e_{m,t+k}|\gamma_{t+k}]]\\
    &=\mathbb{E}[\mathbb{E}[X_{nt}-a_{n}-b_{t}|\gamma_{t}]\mathbb{E}[X_{m,t+k}-a_{m}-b_{t+k}|\gamma_{t+k}]]\\
    &=\mathbb{E}[(b_{t}-b_{t}-\mathbb{E}[a_{n}])(b_{t+k}-b_{t+k}-\mathbb{E}[a_{m}])]\\
    &=0
\end{align*}
Thus, $\mathbb{E}[X_{nt}X_{m,t+k}]=\mathbb{E}[b_{t}b_{t+k}]$ as desired.

\section{Additional Simulation Results}\label{Append_sim}
\subsection{Power}
In this part, we present the simulation result for power of various methods using the model specified in Section \ref{sec4}. The true value of $\beta_{1}$ is still 1, while we test the null hypothesis $H_{0}:\beta_{1}=b_{1}$ for $b_{1}\in[0.5,1.5]$ with a size of $5\%$. We set the sample size $N=T=100$ and run 1000 repetitions for each $\rho\in\{0,0.25,0.5,0.75\}$. Figure \ref{fig_power} shows the power for EHW, CRi, CRt, CGM, CHS, CV and two subsampling variance estimators under different level of dependence. All estimators tend to have higher power as $b_{1}$ is further away from the true value.
\begin{figure}[H]
    \centering
    \includegraphics[width=\textwidth]{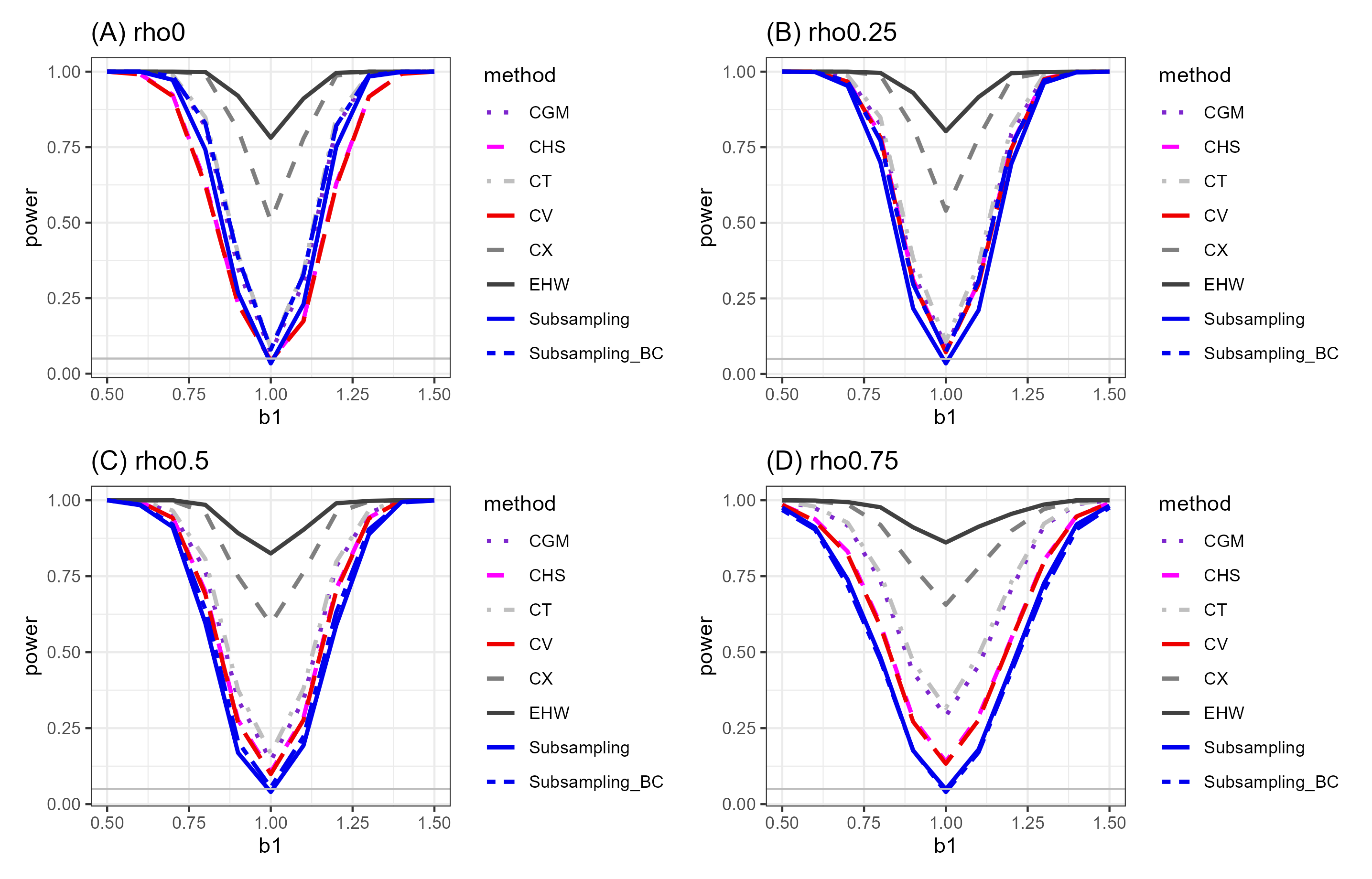}
    \caption{Powers for different null hypothesis with a size of $5\%$ and 1000 Monte Carlo repetitions under various dependence level. The two feasible subsampling variance estimators use the B\"uhlmann's iterative method to select subsample sizes. The sample size is set to be $N=T=100$.}
    \label{fig_power}
\end{figure}

\newpage
\bibliography{references}

\end{document}